\documentclass[12pt]{article}
\usepackage{amsmath}
\usepackage{graphicx}
\usepackage{enumerate}
\usepackage[numbers,sort&compress]{natbib}
\usepackage{url} 

\usepackage{multibib}

\newcites{Supp}{References for the Supplement}

\usepackage{bm, commath}
\usepackage{caption}
\usepackage{graphicx}
\usepackage{subcaption}
\usepackage{amsmath, amsfonts, amsthm}
\usepackage{float}
\usepackage{enumitem}
\usepackage{booktabs,siunitx}
\usepackage{url}
\usepackage{multirow}
\usepackage{amssymb}
\usepackage{bm}
\usepackage{mathrsfs}
\usepackage{authblk}
\usepackage[normalem]{ulem}
\allowdisplaybreaks

\usepackage{listings}
\DeclareMathOperator*{\argmin}{arg\,min}

\usepackage{bbm}
\usepackage{textcomp}
\usepackage{authblk}
\usepackage[ruled]{algorithm2e}
\SetKwInput{KwParam}{Parameter}
\SetAlgoCaptionLayout{centerline}

\usepackage{sectsty}
\usepackage[onehalfspacing]{setspace}

\usepackage[usenames,dvipsnames]{color}
\usepackage[utf8]{inputenc}
\usepackage{colortbl}
\usepackage{float}
\usepackage{tikz}
\usepackage{bbding}
 \usepackage{tabularx}
\usepackage{geometry}
\usepackage{pdflscape}
\usepackage{adjustbox}
\usepackage{threeparttable}

\usepackage{array}

\usepackage{booktabs,tabularx,ragged2e,hyperref}

\usetikzlibrary{trees}
\usetikzlibrary{shapes,decorations,arrows,calc,arrows.meta,fit,positioning}
\usetikzlibrary{shapes.multipart}
\tikzset{
    -Latex,auto,node distance =1 cm and 1 cm,semithick,
    state/.style ={ellipse, draw, minimum width = 0.7 cm},
    state1/.style ={ draw, minimum width = 0.7 cm},
    point/.style = {circle, draw, inner sep=0.04cm,fill,node contents={}},
    bidirected/.style={Latex-Latex,dashed},
    el/.style = {inner sep=2pt, align=left, sloped}
}

\newtheorem{theorem}{Theorem}

\newtheorem{assumption}{Assumption}

\newtheorem{lemma}{Lemma}

\usepackage{mathtools}

\usepackage{xcolor}

\makeatletter
\renewcommand{\algocf@captiontext}[2]{#1\algocf@typo. \AlCapFnt{}#2} 
\def\@algocf@capt@plain{top}
\renewcommand{\algocf@makecaption}[2]{%
  \addtolength{\hsize}{\algomargin}%
  \sbox\@tempboxa{\algocf@captiontext{#1}{#2}}%
  \ifdim\wd\@tempboxa >\hsize
  \hskip .5\algomargin%
  \parbox[t]{\hsize}{\algocf@captiontext{#1}{#2}}
  \else%
  \global\@minipagefalse%
  \hbox to\hsize{\box\@tempboxa}
  \fi%
  \addtolength{\hsize}{-\algomargin}%
}
\makeatother

\definecolor{dblue}{HTML}{0072B2}

\def\bmZ{\bm{Z}}
\def\bmX{\bm{X}}

\def\bmU{\bm{U}}

\def\bgamma{\bm{\gamma}}
\def\bGamma{\bm{\Gamma}}
\def\bbeta{\bm{\beta}}
\def\var{{\rm var}}
\def\cov{{\rm cov}}
\newcommand{\probP}{\text{I\kern-0.15em P}}

\usepackage{etoolbox}
\newcommand{\zerodisplayskips}{%
  \setlength{\abovedisplayskip}{3pt}%
  \setlength{\belowdisplayskip}{3pt}%
  \setlength{\abovedisplayshortskip}{3pt}%
  \setlength{\belowdisplayshortskip}{3pt}}
\appto{\normalsize}{\zerodisplayskips}
\appto{\small}{\zerodisplayskips}
\appto{\footnotesize}{\zerodisplayskips}

\begin{document}

\title{Group Identification and Variable Selection in Multivariable Mendelian
Randomization with Highly-Correlated Exposures}

\author{Yinxiang Wu$^1$, Neil M. Davies$^{2,3,4} $, Ting Ye$^1$}
\affil{\normalsize $^1$Department of Biostatistics, University of Washington,  Seattle, Washington, U.S.A.\\
$^2$ Division of Psychiatry, University College London, London WC1E 6BT, UK.\\
$^3$ Department of Statistical Science, University College London, London WC1E 6BT, UK.\\
$^4$ Department of Public Health and Nursing, Norwegian University of Science and Technology, Norway.
}
\date{}  
\maketitle

\def\spacingset#1{\renewcommand{\baselinestretch}%
{#1}\small\normalsize} \spacingset{1}

\begin{abstract}
Multivariable Mendelian Randomization (MVMR) estimates the direct causal effects of multiple risk factors on an outcome using genetic variants as instruments. The growing availability of summary-level genetic data has created opportunities to apply MVMR in high-dimensional settings with many strongly correlated candidate risk factors. However, existing methods face three major limitations: weak instrument bias, limited interpretability, and the absence of valid post-selection inference. Here we introduce MVMR-PACS, a method that identifies signal-groups---sets of causal risk factors with high genetic correlation or indistinguishable causal effects---and estimates the direct effect of each group. MVMR-PACS minimizes a debiased objective function that reduces weak instrument bias while yielding interpretable estimates with theoretical guarantees for variable selection. We adapt a data-thinning strategy to summary-data MVMR to enable valid post-selection inference. In simulations, MVMR-PACS outperforms existing approaches in both estimation accuracy and variable selection. When applied to 27 lipoprotein subfraction traits and coronary artery disease risk, MVMR-PACS identifies biologically meaningful and robust signal-groups with interpretable direct causal effects. 
\end{abstract}

\noindent%
\vfill

\newpage
\spacingset{1.5} 

\section{Introduction}

Multivariable Mendelian randomization (MVMR) uses genetic variants as instrumental variables to estimate the direct causal effects of multiple exposures on an outcome \citep{Davey-Smith:2003aa, sanderson2022Mendelian}. With the growing availability of summary-level genetic association data, summary-data MVMR has been widely applied across various fields \citep{sanderson2020multivariable,richmond2022mendelian}, including causal inference with omics data \citep{sadler2022quantifying, lu2024integrative, yoshiji2025integrative}. {However, as the number and correlation of exposures increase, current MVMR methods face three critical challenges that limit their reliability: weak instrument bias arising from highly correlated exposures,  difficulty in selecting true causal factors or interpretable groups
among many candidates, and lack of valid post-selection inference. Addressing these challenges is essential for extending MVMR to modern high-dimensional omics applications, which often involve many highly correlated exposures.}


MVMR requires genetic variants (typically single-nucleotide polymorphisms, or SNPs) to satisfy three core assumptions \citep{hernan2006estimating,sanderson2021testing} (Figure \ref{main fig 1}a):
\begin{enumerate}
    \item [i] \textit{Relevance}: SNPs must be strongly associated with each exposure, conditional on the other exposures in the model;
    \item [ii] \textit{Independence}: SNPs are independent of all confounders of the exposures and outcome;
    \item [iii] \textit{Exclusion restriction}: SNPs affect the outcome only through the included exposures.
\end{enumerate}
Unlike univariable MR, which assumes genetic instruments affect the outcome through a single exposure, MVMR incorporates multiple exposures to strengthen the plausibility of the exclusion restriction by adjusting for known pleiotropic pathways. However, this adjustment creates a fundamental trade-off: when exposures are correlated, conditioning on multiple exposures weakens the conditional associations between SNPs and individual exposures. Even when SNPs show strong marginal associations with each exposure---indicating strong instruments in a univariable setting---their conditional associations can become weak when genetic effects on different exposures are highly correlated \citep{sanderson2021testing,patel2024weak,wu2024more}. This problem is particularly acute in omics applications, where genetic correlations between molecular traits such as metabolites or proteins often exceed 0.8 \citep{zhao2021Mendelian,chan2024novel,qiang2025genetic}. In these cases, 
standard estimators such as MVMR-IVW \citep{burgess2015multivariable}, MVMR-Egger \citep{rees2017extending}, and MVMR-Median \citep{grant2021pleiotropy}, as well as recent weak-instrument-robust methods including GRAPPLE \citep{Wang2021grapple}, MVMR-cML \citep{lin2023robust}, MRBEE \citep{lorincz2023mrbee}, and SRIVW \citep{wu2024more}, can yield biased, imprecise, or unstable estimates. Consequently, practitioners often avoid MVMR entirely, opting instead for extensive univariable MR analyses or restricting MVMR to a few preselected exposures---approaches that impose heavy multiple-testing burdens and risk overstating or misattributing causal effects driven by correlated pathways. Scalable methods that can accommodate many correlated exposures while providing valid inference are therefore urgently needed.

Several methods have been proposed to extend MVMR to high-dimensional settings, but each has important limitations. MR-BMA \citep{zuber2020selecting} uses Bayesian model averaging to prioritize likely causal factors but does not correct for weak instrument bias, which can compromise both estimation and selection.  MVMR-cML-SuSiE \citep{chan2024novel}  combines MVMR-cML with SuSiE to cluster exposures but assumes only one exposure per cluster has a direct effect. Both methods also require prespecifying the number of clusters or causal factors, which is typically unknown. Frequentist approaches include sparse principal component analysis \citep{karageorgiou2023sparse}, which constructs principal components from SNP–exposure associations but yields effects that are difficult to interpret biologically. Penalized regression approaches have also been developed. \citet{grant2021pleiotropy} estimate the effect of a primary exposure while applying an $\ell_1$ penalty to select pleiotropic traits. \citet{hao2025transfer} proposed a transfer learning framework that employs the Minimax Concave Penalty (MCP) for variable selection, which penalizes exposures individually but does not account for structured correlations or shared effects among exposures. However, both approaches fail to mitigate weak instrument bias. Critically, none of these methods provide valid post-selection inference, leaving statistical uncertainty poorly quantified after variable selection.

Here we introduce MVMR-PACS, a statistical framework for summary-data MVMR that addresses risk factor selection and post-selection inference in settings with many or highly correlated exposures.  In such settings, disentangling individual causal effects is challenging because conditional instrument strength is limited. MVMR-PACS addresses this through three key innovations. First, it mitigates weak instrument bias using a debiased loss function \citep{wu2024more}. Second, it concentrates signal by adaptively grouping exposures that are genetically correlated or share causal effects through the pairwise absolute clustering and shrinkage (PACS) penalty \citep{sharma2013consistent}. Third, it estimates the causal effect of each group (Figure \ref{main fig 1}c), where the resulting estimand corresponds either to a common direct causal effect within the group or to a weighted average of member effects. For rigorous post-selection inference, we adapt the recently proposed data-thinning approach \citep{neufeld2023data} to summary-data MVMR (Figure \ref{main fig 1}d). Compared with existing approaches, MVMR-PACS provides interpretable grouping of risk factors with theoretical guarantees for both estimation and inference after selection.

We establish that MVMR-PACS satisfies the oracle property for variable selection under reasonable conditions, ensuring consistent identification of true causal factors and signal-groups. Through simulations, we demonstrate its superior performance relative to existing methods in both estimation accuracy and variable selection. We then apply MVMR-PACS to 27 lipoprotein subfraction traits to identify biologically meaningful signal-groups influencing coronary artery disease risk and their direct causal effects.

\section{Results}

\subsection{Overview of MVMR-PACS}

To illustrate the MVMR-PACS workflow, we first provide a schematic overview before presenting simulation and applied analyses. We consider the widely used two-sample summary-data design, where SNP-exposure and SNP-outcome associations are obtained from non-overlapping cohorts. The framework begins by identifying SNPs that satisfy the three core instrument variable assumptions (Figure 1a). Candidate SNPs are selected based on association with at least one exposure at a specified significance threshold (for example, $p<5\times 10^{-8}$), followed by linkage disequilibrium clumping to retain independent variants (Figure 1b). The exposure and outcome datasets are then harmonized to include only these instruments. Details of instrument variable selection and data preprocessing are provided in Methods (Section \ref{sec: iv selection}).

\begin{figure}[ht]
    \centering
    \includegraphics[width=\textwidth]{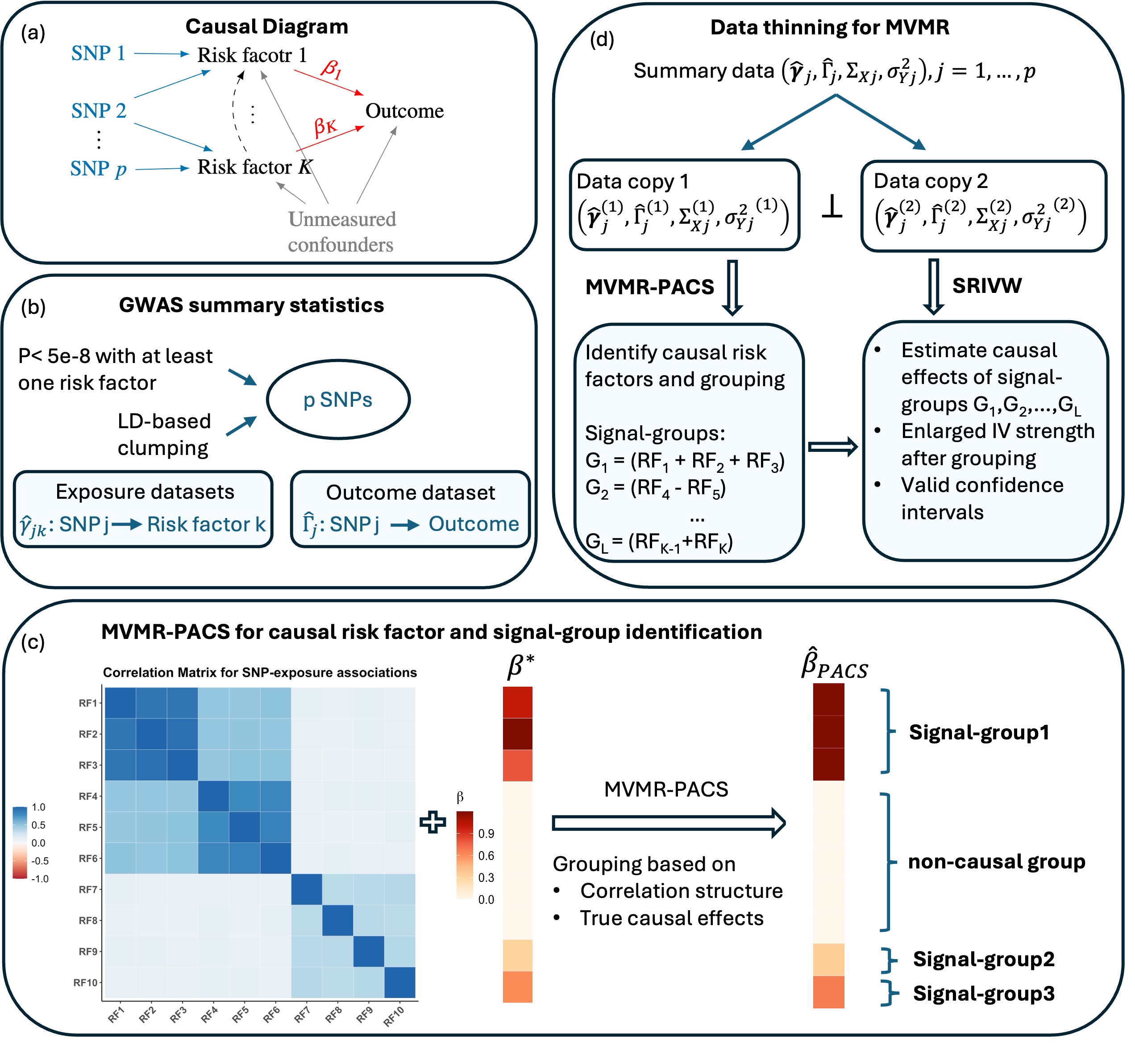}
    \caption{{Overview of the MVMR-PACS framework for variable selection and post-selection inference. \textbf{(a)} We look for SNPs that satisfy the three core instrument assumptions:  SNPs are conditionally associated with each exposure (relevance), are independent of unmeasured confounders (independence), and influence the outcome only through the included exposures (exclusion restriction). \textbf{(b)} Instruments are selected using a p-value threshold and LD clumping, based on GWAS summary statistics. \textbf{(c)} An illustrative example demonstrating the MVMR-PACS approach; MVMR-PACS identifies signal-groups of exposures, based on both correlation structure of observed SNP-exposure associations and true causal effects $\bbeta^*$. \textbf{(d)} Data thinning procedure for summary-data MVMR; observed summary statistics are split into two independent copies, enabling valid post-selection inference. RF: risk factor.}} \label{main fig 1}
\end{figure}

Using the harmonized dataset, MVMR-PACS identifies groups of risk factors that may jointly influence the outcome by combining exposures with high genetic correlation or indistinguishable causal effects. Figure 1c illustrates the approach using a hypothetical example with ten exposures. The correlation matrix of SNP–exposure associations reveals three clusters, though not all clusters contain causal exposures.

A key feature of MVMR-PACS is its ability to form signal-groups that capture both correlation structure and underlying causal effects in a data-driven manner. For instance, as illustrated in Figure \ref{main fig 1}c, {risk factors 1--3} are highly correlated; the method shrinks their estimated effects toward a common value, even if their true effects differ slightly, recognizing that instrument strength is insufficient to distinguish them. {Risk factors 4--6} are also highly correlated but non-causal; their estimated effects are correctly shrunk toward zero. {Risk factors 7--10} exhibit weaker correlations and are treated as separate groups. We define a \emph{signal-group} as a set of exposures assigned a nonzero estimated effect of equal magnitude. This grouping mitigates instability from multicollinearity, overcomes weak instrument bias by concentrating weak signals, and yields interpretable results.

The final step is inference on the identified signal-groups. To ensure valid post-selection inference, we adapt a data thinning strategy \citep{neufeld2023data}. {For each selected SNP, the vector of SNP–exposure associations is carefully perturbed to yield two independent components that preserve distributional properties of the original vector; SNP–outcome associations are treated analogously (Method Section \ref{sec: data thin}). This procedure yields two independent datasets containing the same set of SNPs. One dataset is used for selection, where MVMR-PACS identifies signal-groups, and the other for inference, where direct effects and confidence intervals are estimated (Figure 1d). Unlike conventional sample-splitting approaches that partition SNPs into disjoint subsets \citep{grant2022efficient, zhao2023robust}, data-thinning preserves all instruments in both datasets. This avoids instability when the number of instruments is limited while accommodating heterogeneous SNP-level association distributions.}

The causal effect of a signal-group has a clear interpretation. When group members are highly genetically correlated, the estimate reflects a weighted average of their individual effects. When weakly or moderately correlated members are grouped together, this indicates that the risk factors exert similar direct causal effects on the outcome, and the estimate reflects their common effect (Methods Section \ref{sec: interpretation}).


\subsection{Simulation results}\label{sec: simulation results}

The simulation study evaluated the finite-sample performance of MVMR-PACS relative to existing MVMR estimators in terms of estimation accuracy and variable-selection performance. We also assessed whether the proposed data-thinning procedure yields valid post-selection inference. We generated summary-level data for ten risk factors under varying sample sizes to assess robustness to weak instruments. Results are summarized using 1,000 Monte Carlo replicates for each scenario.

Each dataset comprised ten risk factors (RF1–RF10) with true causal effects of $(1, 1, 1, 0, 0, 0, 0, 0, 0.5, 0)$, respectively. Summary-level SNP–exposure and SNP–outcome associations were simulated from individual-level models at sample sizes $N = 1\times10^5$, $2\times10^5$, and $3\times10^5$. The correlation structure of SNP–exposure associations mimicked the block pattern shown in Figure 1c (Supplementary Figure S1), creating correlated exposures with conditionally weak instruments. Further details of the data-generating process are provided in Methods (Section \ref{sec: sim details}).

The target estimand for each method was the direct causal effect of each exposure on the outcome. We compared MVMR-PACS and its variants---MVMR-PACS-0.8, {which incorporates a correlation threshold of 0.8 to disable grouping between weakly correlated exposures}, and MVMR-dLASSO, corresponding to the special case where the PACS penalty reduces to an $\ell_1$ penalty---against several established estimators listed in Table \ref{tab:methods_summary}.

\begin{table}[ht]
\centering
\caption{MVMR estimators compared in simulation and real-data analyses. Selection rules specify criteria for identifying exposures with nonzero direct causal effects on the outcome.}
\label{tab:methods_summary}
\footnotesize
\resizebox{\textwidth}{!}{
\begin{tabular}{p{3.3cm} p{8.5cm} p{4.5cm}}
\toprule
\textbf{Method} & \textbf{Overview} & \textbf{Selection rule} \\
\midrule
\textbf{MVMR-IVW} \citep{burgess2015multivariable} &
Inverse-variance weighted estimator assuming all instruments are valid and sufficiently strong. &
Bonferroni-corrected $p < 0.05/K$ ($K$ exposures). \\[3pt]
\addlinespace[4pt]
\textbf{SRIVW} \citep{wu2024more} &
Weak-instrument–robust extension of MVMR-IVW. &
Bonferroni-corrected $p < 0.05/K$ ($K$ exposures). \\[3pt]
\addlinespace[4pt]
\textbf{IVW-LASSO} (implemented in \cite{zuber2020selecting}) &
Adds an $\ell_1$ penalty to MVMR-IVW for sparse estimation. &
Selected if nonzero\textsuperscript{1} after penalization. \\[3pt]
\addlinespace[4pt]
\textbf{MR-BMA} \citep{zuber2020selecting} &
Bayesian model averaging exploring likely exposure subsets; computes marginal inclusion probabilities (MIPs). &
Selected if MIP $>0.2$, $0.5$, or $0.8$ (default: $>0.5$). \\[3pt]
\addlinespace[4pt]
\textbf{MVMR-cML-SuSiE} \citep{chan2024novel} &
Combines constrained maximum likelihood with SuSiE to identify causal clusters and estimate one effect per cluster. &
Selected if exposure belongs to any causal cluster. \\
\addlinespace[4pt]
\textbf{MVMR-PACS} &
Adaptively group highly correlated exposures or those with indistinguishable causal effects and estimate each group's direct effect. &
Selected if exposure belongs to any signal-group\textsuperscript{1}. \\
\addlinespace[4pt]
\textbf{MVMR-PACS-0.8} &
MVMR-PACS with a correlation threshold of 0.8; disabled grouping between exposures whose SNP-exposure associations have absolute correlation less than 0.8. &
Selected if exposure belongs to any signal-group\textsuperscript{1}. \\
\addlinespace[4pt]
\textbf{MVMR-dLASSO} &
MVMR-PACS with grouping disabled, i.e., when the PACS penalty reduces to a usual $\ell_1$ penalty. &
Selected if nonzero\textsuperscript{1} after penalization. \\
\bottomrule
\end{tabular}}
\vspace{0.1em}
\begin{flushleft}
\footnotesize\textit{\textsuperscript{1}:} Absolute values of estimated effects greater than 0.001.
\end{flushleft}
\end{table}

Method performance was summarized using: (1) \textit{Estimation accuracy}: median mean-squared error (MSE) across replicates,  defined as the squared difference between estimated and true causal effects; (2) \textit{Variable selection performance}: correct sparsity (proportion of exposures correctly identified as having zero or nonzero effects), sensitivity, and false-positive rate, accounting for direction of effect. To evaluate post-selection inference, we computed empirical coverage probability of nominal 95\% confidence intervals on the inference dataset, conditional on signal-groups identified in the selection dataset.

\subsubsection{Estimation accuracy and risk factor selection}

Figure \ref{main sim fig} summarizes results across 1,000 replicates; numerical results are presented in Supplementary Table S1. At $N=1\times10^5$, overall instrument strength was weak, with a mean instrument strength parameter \citep{wu2024more} of -0.2  and conditional $F$-statistics \citep{sanderson2021testing} ranging from 1.1 to 13.1. The first three risk factors had mean conditional $F$-statistics of 1.1, indicating extremely weak conditional instruments.

In this setting, MVMR-PACS achieved the lowest MSE and the most stable point estimates (Supplementary Figure S2).
MR-BMA ranked second but frequently produced negative estimates for null effects (RF5, RF7, RF8, RF10), suggesting directional bias under weak instruments. MSE was not reported for MVMR-cML-SuSiE, as it evaluates all possible combinations of cluster representatives rather than producing a single set of exposure effects.

For variable selection, MVMR-PACS achieved the highest correct sparsity (88.3\%), outperforming MR-BMA (75.3\%), MVMR-cML-SuSiE (73.8\%), and other methods.
Its sensitivity reached 83.5\%, compared with 86.2\% for MR-BMA and 45.1\% for MVMR-cML-SuSiE, while maintaining a low false-positive rate (8.5\% versus 31.9\% and 7.1\%, respectively).
False positives from MVMR-PACS mainly reflected moderate correlations between RF1–RF3 and RF4–RF6, where non-causal exposures (RF4–RF6) occasionally joined the same signal-group.
When a correlation threshold (MVMR-PACS-0.8) was imposed, this grouping was prevented, yielding improved correct sparsity and reduced false-positive rates (Figure \ref{main sim fig}).
These findings demonstrate how incorporating prior knowledge of risk factor group structure can enhance MVMR-PACS performance in
practice.

IVW-LASSO achieved correct sparsity comparable to MVMR-PACS and MR-BMA but exhibited much larger MSE and lower sensitivity, especially at small sample sizes.
MVMR-dLASSO showed even lower sensitivity, often selecting only one or two risk factors from RF1–RF3. SRIVW produced nearly unbiased estimates but with high variance, whereas MVMR-IVW displayed substantial bias and variability (Supplementary Figure S2). As expected, Bonferroni correction yielded conservative selection, with low sensitivity.

\begin{figure}[ht]
\centering\includegraphics[width=\textwidth]{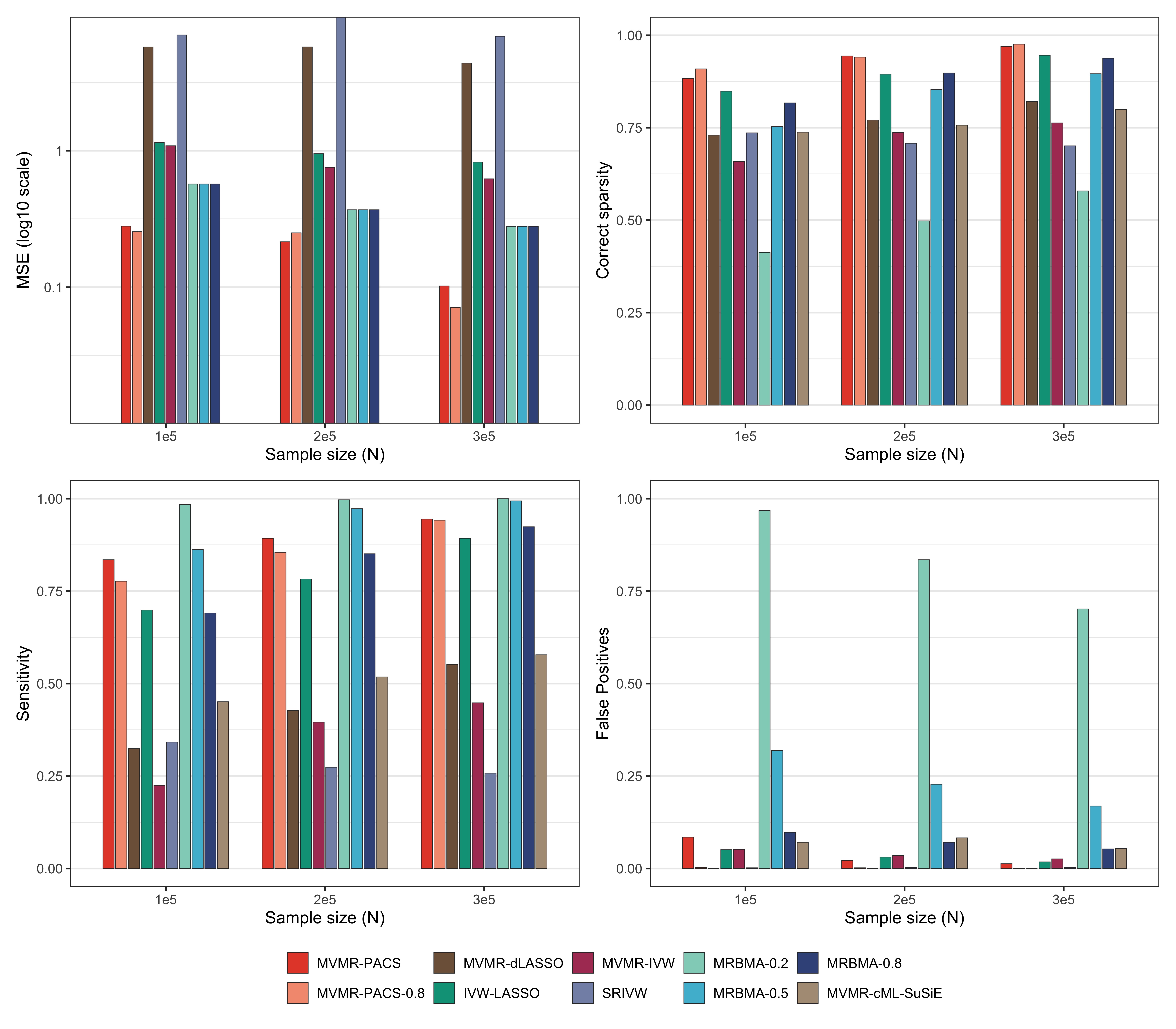}
    \caption{{Performance of MVMR estimators in causal effect estimation and variable selection (Table \ref{tab:methods_summary}). The figure summarizes results from simulations with sample sizes $N=1\times10^5$, $2\times10^5$, and $3\times10^5$, reporting median mean squared error (MSE), correct sparsity, sensitivity, and false positive rate across 1,000 replicates. The average instrument strength parameters are -0.2, 1.9, and 4.0, for each sample size respectively. MVMR-PACS-0.8 is a variant of MVMR-PACS, incorporating a correlation threshold of 0.8 into its penalty. MR-BMA-0.2, MR-BMA-0.5, and MR-BMA-0.8 denote Bayesian model averaging with MIP thresholds of 0.2, 0.5, and 0.8; they all share the same model-averaged effect estimates and hence the same MSE estimates. MSE is not reported for MVMR-cML-SuSiE.}} \label{main sim fig}
\end{figure}

As sample size increased to $2\times10^5$ and $3\times10^5$, performance improved for all estimators, with MVMR-PACS achieving the lowest MSE, highest correct sparsity, and highest sensitivity while maintaining a comparable false positive rate.

\subsubsection{Post-selection Inference}

We next examined whether the proposed data-thinning procedure yields valid confidence intervals conditional on the selection event.
Using the same data-generation setup, we conducted 1,000 replications at each sample size $N=1\times10^5$, $2\times10^5$, and $3\times10^5$. In each repetition, observed SNP-exposure and SNP-outcome associations were carefully perturbed to yield two independent datasets: a selection dataset, used by MVMR-PACS to identify signal groups, and an inference dataset, used for post-selection inference (Figure~\ref{main fig 1}; Methods Section~4.6). Confidence intervals for the direct effects of the selected signal-groups were then constructed using the SRIVW estimator and evaluated against the true effects. In this setting, two signal-groups existed by design: one comprising RF1–RF3, which share the same true causal effect of 1, and another comprising RF9, with a true causal effect of 0.5.

Due to added noise and reduced instruments' strength in the selection dataset, MVMR-PACS exhibited slightly weaker performance than with full data (Supplementary Table S2). At sample size $N=1\times10^5$, the most frequently observed grouping placed RF1–RF3 in a signal group, with remaining risk factors identified as non-causal. This grouping occurred in 775 of 1,000 runs. The number of runs in which MVMR-PACS correctly identified all causal risk factors increased with sample size: 132 at $N=1\times10^5$, 342 at $N=2\times10^5$, and 442 at $N=3\times10^5$. The complete distribution of grouping frequencies is shown in Supplementary Table S3.

Coverage probabilities of SRIVW-based confidence intervals constructed on the inference set were close to the nominal 0.95 level when the correct model was selected (Supplementary Table S4).

\subsection{Lipoprotein Subfraction Traits as Risk Factors for Coronary Artery Disease}

\subsubsection{Signal-groups identified by MVMR-PACS}

We applied MVMR-PACS to traditional lipid profiles (HDL-C, LDL-C, TG, ApoA1, ApoB) and lipoprotein subfraction traits to identify signal-groups that may jointly influence coronary artery disease (CAD). We also applied the proposed data thinning technique to estimate group effects and construct confidence intervals.

Lipoprotein subfractions and particle sizes have gained increasing attention in cardiovascular risk assessment \citep{mora2009advanced,emerging2012lipid,chary2023association,lawler2017residual}, yet observational studies have reported conflicting findings \citep{zhao2021Mendelian}. To address this, \citet{zhao2021Mendelian} conducted an MR study of 82 traits spanning VLDL, LDL, IDL, and HDL subfractions and particle sizes. Each subfraction trait is named using three components separated by hyphens: size (XS, S, M, L, XL, XXL), lipoprotein density (VLDL, LDL, IDL, HDL), and measurement (C for total cholesterol, CE for cholesterol esters, FC for free cholesterol, L for total lipids, P for particle concentration, PL for phospholipids, TG for triglycerides). For example, M-HDL-P denotes the concentration of medium HDL particles. In addition, VLDL-D, LDL-D, and HDL-D represent the average diameters of the corresponding lipoprotein classes.
After excluding traits with genetic correlations above 0.8 with traditional lipids, they retained 27 subfraction traits that may capture independent biological mechanisms, yet substantial correlations remained in the SNP–exposure associations. Consequently, each trait was analyzed separately while adjusting for traditional lipids, an approach that carries a high multiple testing burden and risks bias from uncontrolled pleiotropy since it is difficult to identify variants that affect a single subfraction without influencing others.

To overcome these limitations, we jointly analyzed the 27 subfraction traits and five traditional lipids using MVMR-PACS. Publicly available GWAS summary statistics were harmonized and inverse rank normalized to place all traits on the same scale (see Supplementary Section 2 for data sources). SNPs were first selected based on genome-wide significance ($p<5\times 10^{-8}$) for association with at least one of the 32 lipid traits.
Variants meeting this criterion were then clumped to ensure independence, using pairwise linkage disequilibrium ($r^2<0.001$) within a 10 Mb window and the European reference panel from the 1000 Genomes Project. Within each linkage disequilibrium block, the SNP showing the smallest p-value for association with any of the lipid traits was retained. This procedure yielded a total of 591 uncorrelated SNPs. The correlation matrix of SNP–exposure associations (Figure  \ref{lipid cor mat}) showed a clear clustered structure. Conditional F-statistics for the 32 traits ranged from 0.16 to 1.31, and the instrument strength parameter was negative (–24.3), both indicating extremely weak instrument strength.

\begin{figure}[ht]
    \centering
    \includegraphics[width=14cm]{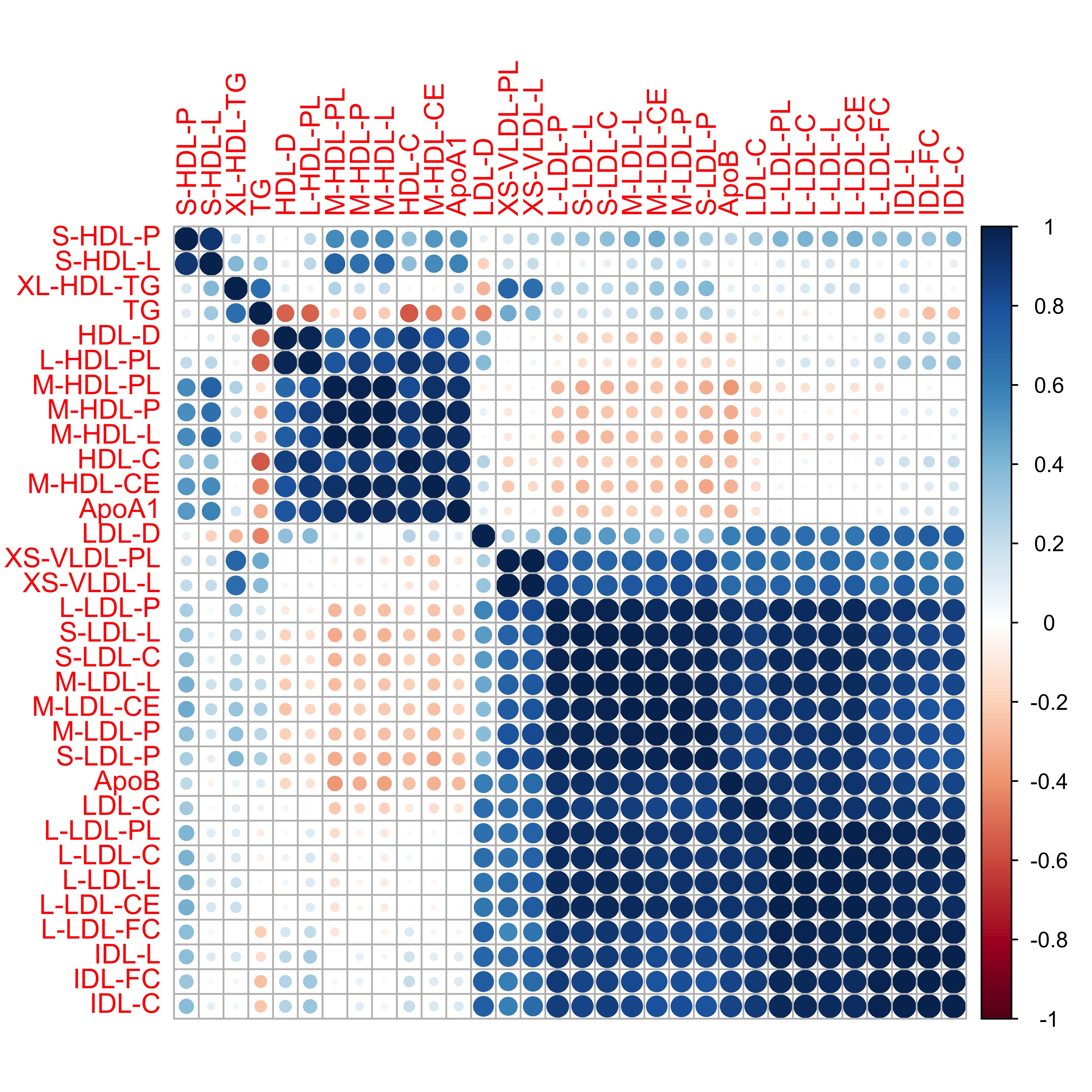}
    \caption{{Correlations of SNP-exposure associations among 32 lipid traits based on 591 selected SNPs.}}\label{lipid cor mat}
\end{figure}

{MVMR-PACS identified seven signal groups (Table \ref{lipid main table}), labeled by their members: TG; LDL-C; ApoB; HDL-D; medium/large HDL subfractions (M-HDL-CE, M-HDL-L, M-HDL-P, M-HDL-PL, L-HDL-PL); LDL subfractions (S-LDL-C, S-LDL-L, S-LDL-P, M-LDL-CE, M-LDL-L, M-LDL-P, L-LDL-P); and a mixed VLDL/IDL/LDL group (XS-VLDL-L, XS-VLDL-PL, IDL-C, IDL-FC, IDL-L, L-LDL-FC, LDL-D). These groupings largely mirrored the correlation-based clusters yet also revealed separations driven by estimated causal effects. For example, despite their high correlation, ApoB and LDL-C formed distinct groups, whereas LDL-D clustered with VLDL/IDL lipid components despite weaker phenotypic correlation.}

{Based on the MVMR-PACS estimates, medium/large HDL subfractions, LDL subfractions, and LDL and HDL diameters showed protective effects on CAD risk, while LDL-C, ApoB, TG, VLDL/IDL subfractions, and L-LDL-FC indicated harmful effects. For groups with strong internal correlation (for example, the signal-group of LDL subfractions), estimates should be interpreted as weighted averages of member traits' effects; for signal-groups with weaker internal correlation (such as LDL-D with VLDL lipid components), the equal magnitude of effects suggests comparably sized causal effects (see Methods Section \ref{sec: interpretation}).} 

\subsubsection{Other methods}

For comparison, we applied six alternative MVMR estimators to the same dataset: MVMR-IVW, SRIVW, IVW-LASSO, MVMR-dLASSO, MR-BMA, and MVMR-cML-SuSiE (Table \ref{lipid main table}).

Estimates from MVMR-IVW and SRIVW were unstable, reflecting the extremely weak instrument strength. IVW-LASSO and MVMR-dLASSO agreed on LDL-C, TG and S-LDL-C, but disagreed on a few correlated traits.

For MR-BMA, the leading risk factors ranked by marginal inclusion probabilities (MIPs) were LDL-C (MIP = 1), TG (MIP = 0.996), ApoB (MIP = 0.989), L-LDL-P (MIP = 0.988), HDL-D (MIP = 0.975), and L-HDL-PL (MIP = 0.967) (Supplementary Table S6). Notably, MR-BMA produced a model-averaged causal effect estimate for ApoB of –0.497, which likely reflects weak instrument bias and contradicts the established consensus that elevated ApoB increases CAD risk \citep{sniderman2019apolipoprotein}.

{MVMR-cML-SuSiE identified three clusters: cluster 1 (LDL-C), cluster 2 (ApoA1 and HDL-C), and cluster 3 (L-LDL-FC, IDL-L, IDL-FC, IDL-C). A representative risk factor from each cluster was then selected, yielding eight possible models. MVMR-cML was applied to each model to estimate causal effects (numerical results are presented in Supplementary Table S7). LDL-C and the cluster of ApoA1 and HDL-C consistently showed harmful effects, whereas risk factors in clusters 3 suggested protective effects. However, despite accounting for collinearity through cluster selection, the point estimates are likely biased by residual pleiotropy, as only one exposure from each cluster is retained for analysis. For example, LDL-C exhibited an estimated effect size of approximately 2.2, corresponding to an 9-fold increase in CAD risk per standard deviation increase in LDL-C levels, far exceeding estimates from other approaches.}

\begin{table}[h]
\centering
\caption{Point estimates of individual traits across estimators. Estimates are shown on the log-odds scale. Values with absolute magnitude below 0.001 are denoted as ``-''. For MR-BMA, model-averaged effect estimates are reported.} \label{lipid main table}
\centering
\resizebox{\ifdim\width>\linewidth\linewidth\else\width\fi}{!}{
\begin{tabular}[t]{ccccccc}
\toprule
Traits & IVW & SRIVW & IVW-LASSO & MVMR-dLASSO & MVMR-PACS & MR-BMA\\
\midrule
XS-VLDL-PL &   7.237 &   1.256 &      - &        - &         0.165 &   0.563 \\
XS-VLDL-L &  -5.663 &  -1.797 &      - &        - &         0.165 &   0.050 \\
L-LDL-FC &   7.825 & -25.623 &      - &        - &         0.165 &  -0.014 \\
IDL-L &  -9.148 &  -9.736 &      - &        - &         0.165 &  -0.308 \\
IDL-FC &  -1.217 &  -0.949 &      - &        - &         0.165 &   0.016 \\
IDL-C &   5.688 &  11.300 &      - &        - &         0.165 &  -0.392 \\
LDL-D &  -0.360 &   1.593 &      - &        - &        -0.165 &  -0.028 \\
L-LDL-P &  -2.084 &  -0.582 &      - &        - &        -0.351 &  -1.328 \\
M-LDL-P &  -1.198 &   1.671 &      - &        - &        -0.351 &   0.020 \\
M-LDL-L &   0.320 & -10.824 &      - &        - &        -0.351 &   0.222 \\
M-LDL-CE &   1.534 &   3.154 &      - &        - &        -0.351 &   0.012 \\
S-LDL-P &  -0.780 &   0.690 &      - &        - &        -0.351 &  -0.803 \\
S-LDL-L &  -3.090 &  -0.501 &      - &        - &        -0.351 &  -0.051 \\
S-LDL-C &   0.586 &   6.723 &     -0.340 &       -1.911 &        -0.351 &  -0.510 \\
L-HDL-PL &   3.782 &   1.797 &      - &        - &        -0.049 &   1.137 \\
M-HDL-PL &  -1.859 &  -7.468 &      - &        - &        -0.049 &  -0.041 \\
M-HDL-P &  -2.633 &   3.198 &      -0.070 &        - &        -0.049 &  -0.213 \\
M-HDL-L &   8.745 &   7.439 &      - &        - &        -0.049 &  -0.065 \\
M-HDL-CE &  -2.902 &  -5.299 &     -0.040 &        - &        -0.049 &   0.012 \\
HDL-D &  -3.580 &  -1.731 &      - &       -0.240 &        -0.314 &  -1.173 \\
ApoB &  -0.752 &   0.784 &      - &        0.721 &         0.422 &  -0.497 \\
LDL-C &   1.307 &   0.637 &      0.732 &        1.301 &         1.218 &   1.201 \\
TG &   0.350 &  -0.392 &      0.117 &        0.188 &         0.362 &   0.408 \\
L-LDL-PL &   7.101 &  -0.181 &      - &        - &         - &   0.134 \\
L-LDL-L &   2.143 &   8.681 &      - &        - &         - &   0.661 \\
L-LDL-C & -38.804 &  71.081 &      - &        - &         - &   0.510 \\
L-LDL-CE &  30.783 & -56.568 &      - &        - &         - &   0.984 \\
XL-HDL-TG &  -0.107 &  -0.062 &      - &        - &         - &  -0.009 \\
S-HDL-P &   0.443 &  -3.357 &     -0.030 &        - &         - &  -0.092 \\
S-HDL-L &  -4.019 &   3.812 &      - &        - &         - &  -0.501 \\
ApoA1 &  -0.263 &   1.634 &     -0.010 &        - &         - &  -0.167 \\
HDL-C &  -0.007 &  -0.409 &     -0.118 &        - &         - &  -0.003 \\
\bottomrule
\end{tabular}}
\end{table}

\subsubsection{Inference with data thinning}

We next applied the proposed data-thinning procedure, which splits the data into two independent copies, one for variable selection and the other for estimation and inference. As with many resampling approaches, random variation can affect results, particularly in settings with weak and noisy signals. To stabilize inference, we repeated the data thinning procedure 100 times.

Figure \ref{heatmap inference} shows a combined heatmap and dendrogram summarizing MVMR-PACS estimates across 100 data thinning replicates. The heatmap displays effects for the 32 traits that were nonzero and statistically significant at $p<0.05$. Across runs, LDL-C, TG, XS-VLDL-PL, and XS-VLDL-L consistently showed harmful effects, while small and medium LDL subfractions and medium HDL subfractions showed protective effects, with concordant significance in more than half of the replicates. These directions also matched the primary full-data results (Table \ref{lipid main table}). Other traits were either not significant in most runs or showed inconsistent directions, indicating limited and unstable evidence of causal effects. 

The dendrogram visualizes clustering of the 32 traits using a custom distance metric defined by the frequency of co-assignment to the same signal-group across runs. Traits that repeatedly grouped together clustered tightly. Several signal-groups from the full-data analysis, including medium HDL subfractions, small and medium LDL subfractions, VLDL lipid components were frequently recovered. In contrast, LDL-C and TG were often isolated into distinct groups, consistent with the full-data analysis in Table \ref{lipid main table}.

\begin{figure}[ht]
    \centering
    \includegraphics[width=\textwidth]{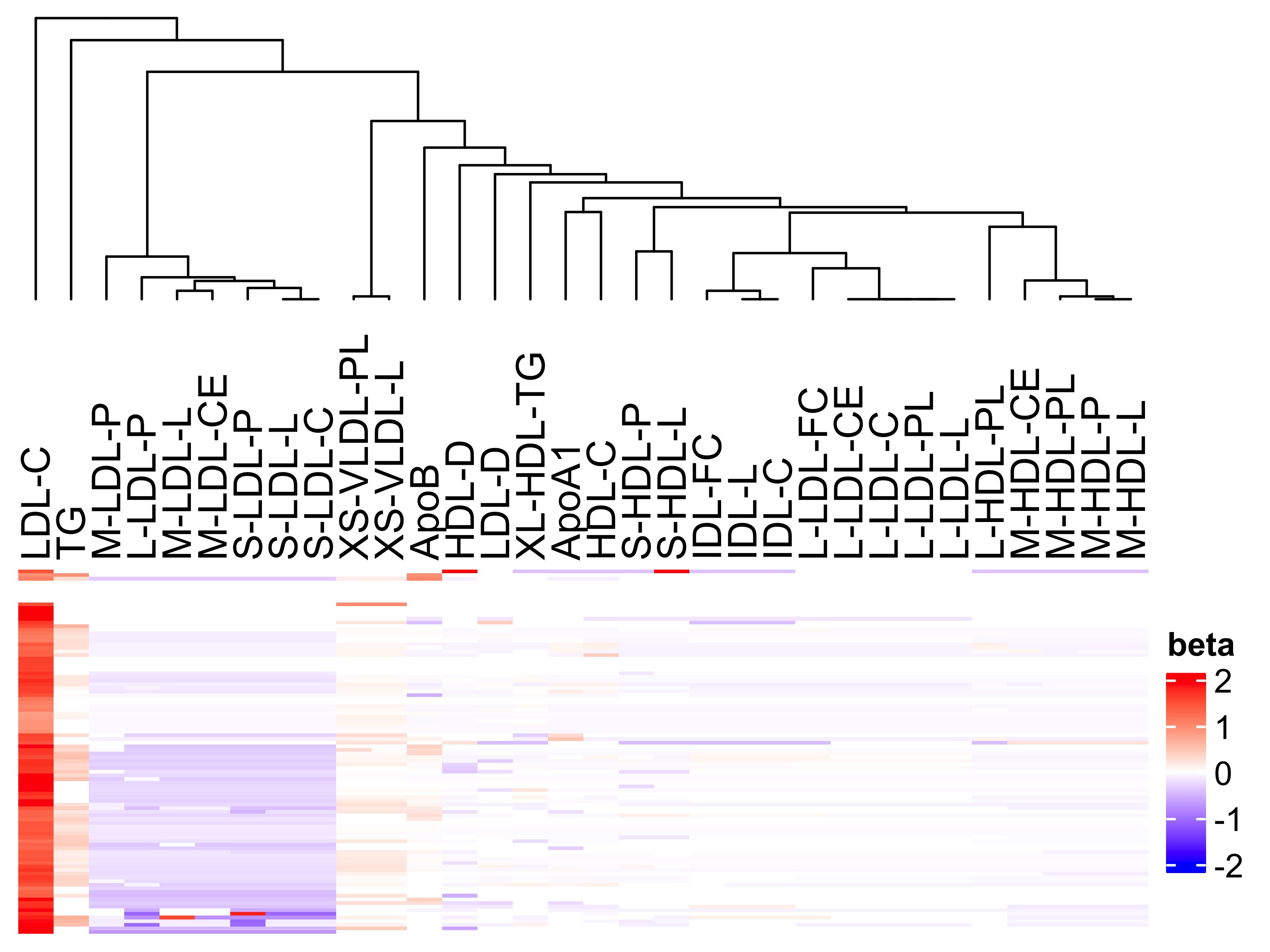}
    \caption{{Heatmap of estimated causal effects across 100 data-thinning replicates. Shown are MVMR-PACS estimates (log-odds scale) for exposures that were nonzero and statistically significant at $p<0.05$ in at least one replicate. Red indicates harmful effects and blue indicates protective effects on CAD risk. The dendrogram represents hierarchical clustering of the 32 traits based on their frequency of co-assignment to the same signal-group across replicates. Median [inter-quartile range] of instrument strength parameter \cite{wu2024more} in the inference dataset is 14.8 [7.5, 49.8] across 100 data thinning replicates.}} \label{heatmap inference}
\end{figure}

\section{Discussion}

We introduced MVMR-PACS, a general framework for risk factor selection and post-selection inference in summary-data MVMR. The method expands the current toolbox to settings with many, often highly correlated, risk factors. By minimizing a debiased objective function that accounts for the clustered structure of SNP–exposure associations, MVMR-PACS provides principled estimation and transparent grouping. It adaptively identifies clusters of risk factors with high correlation or indistinguishable effects, and interprets the direct effect of a signal-group as a weighted average of its members’ causal effects. 
When risk factors are not highly correlated, or when there is sufficient instrument strength to distinguish their individual effects, MVMR-PACS treats them as distinct signal groups. This adaptive, data-driven behavior makes MVMR-PACS broadly applicable. We formally demonstrate that MVMR-PACS exhibits the oracle property of variable selection: it can consistently identify causal risk factors in the presence of many weak instruments (Methods, Section \ref{sec: theory}).

To address the challenge of inference following risk factor selection, we extended the data thinning approach to summary-data MVMR. This approach generates independent analytic datasets of summary statistics that retain the same number of SNPs as the original dataset while preserving key information necessary for both causal effect estimation and signal group identification. By using one dataset for signal group identification and an independent dataset for inference, our method yields valid post-selection inference when the true model is correctly identified. Moreover, this approach is broadly applicable to other tasks in MVMR studies, such as cross-validation, where traditional sample-splitting based on available instruments may be infeasible or inappropriate. Therefore, it serves as a valuable tool that complements existing methods in the MVMR literature.

Our method has several limitations. First, signal groups identified by MVMR-PACS may comprise heterogeneous components without feasible intervention targets to modify the entire group. We therefore advise against interpreting the estimated direct effect of a signal group in the traditional way---{as the effect of a unit change in the exposure on the outcome}. Instead, the effect should be viewed as a weighted average of the direct causal effects of the risk factors within the group. A nonzero group effect indicates the presence of at least one true causal risk factor, but disentangling individual contributions will likely require new data and methods beyond MR.

Second, although Theorem \ref{theorem: PACS} establishes that MVMR-PACS enjoys the oracle property for risk factor selection in the presence of many weak instruments, this guarantee relies on specific conditions. A sufficient condition is that, as sample size grows, the ratio $\mu_{n,\min}/\sqrt{p}\to \infty$, where $p$ is the number of instruments and $\mu_{n,\min}/\sqrt{p}$ is the instrument strength parameter defined in \citet{wu2024more}. {This parameter reflects the weakest instrument strength across all linear combinations of exposures and serves as a diagnostic measure for estimators such as MVMR-IVW or SRIVW that aim to estimate each exposure’s direct effect. In practice, one can compute its empirical value, $\hat{\mu}_{n,\min}/\sqrt{p}$; values smaller than the rule-of-thumb threshold of 7, or negative values, indicate potential weak-instrument bias, as observed in our data application. However, this does not necessarily imply failure of MVMR-PACS. Unlike estimators that attempt to recover all individual exposure effects, MVMR-PACS adaptively aggregates highly correlated exposures or those with indistinguishable effects into signal groups, thereby mitigating instability from extremely weak instruments while still producing accurate estimates for groups with adequate conditional instrument strength. Additionally, our simulation study shows that MVMR-PACS maintains low MSE and good variable selection performance even when $\hat{\mu}_{n,\min}/\sqrt{p} < 0$, indicating that the instrument strength condition is sufficient but not necessary for satisfactory finite-sample performance. These observations suggest that the formal instrument-strength condition may be stronger than necessary for MVMR-PACS, since its practical goal in settings with highly correlated exposures is not to disentangle every individual effect but to identify and estimate interpretable signal groups.} Relaxing this requirement remains an important direction for future work.

{Third, as with any variable-selection method, MVMR-PACS may fail to correctly identify all signal-groups. This limitation reflects a fundamental challenge rather than a specific shortcoming of the method: when the available instruments contain insufficient independent information, no data-driven approach can fully separate highly correlated causal signals. Nevertheless, MVMR-PACS provides a transparent framework that allows users to incorporate prior knowledge to improve stability. The PACS penalty can be customized by setting selected pairwise weights to zero to prevent grouping between specific traits, or by assigning large positive weights to enforce grouping. As shown in Section~2.3, adjusting the correlation threshold is another way to influence the resulting group structure. This flexibility enables domain expertise to guide the grouping process, ensuring that the resulting signal-groups reflect both empirical and biological evidence rather than relying solely on data-driven process.}

{In practice, one can diagnose whether MVMR-PACS has likely identified reliable signal-groups through reasonable perturbations of data, such as repeated data-thinning runs or bootstrap resamples.  Stable group membership indicates robust clustering. Conversely, highly variable group composition across runs suggests uncertainty due to weak instruments or collinearity.}

Fourth, data thinning, used for both tuning parameter selection in MVMR-PACS and post-selection inference, introduces randomness when generating independent copies of the dataset, which may affect performance and signal-group identification. To ensure stability and reproducibility, we recommend repeating the procedure multiple times. For example, in tuning parameter selection, a five-fold data thinning procedure can be repeated several times, and the tuning parameter that minimizes the average loss across runs selected. This process can be parallelized to reduce computation time.

Summarizing post-selection inference across repeated runs is more nuanced, as signal-groups may vary slightly between iterations and the literature provides little guidance on how to integrate such results. We explored two strategies. The first is to focus on the most frequently observed grouping and compute the median of point estimates and standard errors for each signal-group across runs, from which confidence intervals and p-values are derived. A caveat is that datasets are independent within runs but not across runs. In addition, when instrument strength is very weak, signal-group identification may be inconsistent across repeated runs, and the most frequent grouping may not represent the most plausible structure. The second strategy, which we recommend in practice, is to descriptively summarize variation in groupings and inference results, for example using a heatmap as in Figure \ref{heatmap inference}. A similar approach was adopted by \citet{dharamshi2025generalized} to assess the stability of changepoint detection with data thinning.

Fifth, the confidence intervals constructed using our data thinning procedure are valid only if MVMR-PACS correctly identifies all signal-groups. Our theoretical results show that MVMR-PACS can consistently recover the true signal-groups under appropriate conditions (Methods, Section \ref{sec: theory}), a finding further supported by our simulations. These results suggest that the confidence intervals are asymptotically valid. In finite samples, however, failure to identify the correct model can lead to misspecification in the inference stage, resulting in invalid confidence intervals.

In summary, MVMR-PACS extends the methodological toolkit for high-dimensional MR by integrating principled risk factor selection with valid post-selection inference. By explicitly accounting for correlation among exposures and addressing the challenges posed by weak instruments, the framework provides a robust approach for disentangling complex causal relationships, particularly in genomic and molecular data.

\section{Methods}
\label{sec: method}
\vspace{-4mm}

\subsection{Model and identification} \label{sec: model and iden}

We consider the problem of estimating the causal effects of $K$ potentially correlated exposures $\bmX = (X_1,...,X_K)^T $ on an outcome $Y$, in the presence of unmeasured confounders $\bmU$. Let $\bmZ = (Z_1,\ldots,Z_p)^T$ denote $p$ independent SNPs, obtained through linkage disequilibrium (LD) pruning or clumping \citep{Hemani:2018ab}, which satisfy the three instrument assumptions. The causal model is given by
\begin{align}
    \Gamma_j =\gamma_{j1} \beta_{1}^* + \gamma_{j2} \beta_{2}^* + \dots + \gamma_{jK} \beta_{K}^*, \quad j=1,\dots, p \label{eq: true causal}
\end{align}
where $\Gamma_j$ is the marginal SNP–outcome association, $\gamma_{jk}$ is the marginal SNP–exposure association for exposure $k$, and $\beta_k^*$ is the direct effect of exposure $X_k$ on $Y$. This formulation arises from the structural equations
\begin{align}
	  &  X_k  = \sum_{j=1}^p Z_j \gamma_{jk} + f_k(\bm U, E_{Xk}),  \quad k=1,\dots, K, \label{eq: exp}  \\ 
        & \text{Var}(X_k) = 1, \quad k=1,\dots, K, \\
	  &  Y = \ \bm X^T \bbeta^* + g(\bm U, E_Y) ,\label{eq: out}
\end{align}
where $\bm\beta^* = (\beta_1^*,\beta_2^*,...,\beta_K^*)^T$, $E_Y$ is a random noise term,  $(E_Y, \bmU)$ is independent of $\bm Z$ by the second instrument assumption,
and $f_1,\dots, f_K, g$ are unspecified functions. We scale the exposures to unit variance so that the elements of $\bbeta^*$ are comparable. The causal model in \eqref{eq: true causal} remains approximately correct even when model \eqref{eq: out} is nonlinear, for example when $Y$ is binary or when $\bmX$ has nonlinear effects on $Y$ \citep{zhao2018statistical,Wang2021grapple}. In such cases, $\bm\beta^*$ represents average local effects of the exposures. From \eqref{eq: true causal}, $\bbeta^*$ is uniquely identified provided that the design matrix
\begin{align}
    \Pi = \bigg[\begin{matrix}
\gamma_{11} & \dots & \gamma_{1K}\\
\vdots & \ddots & \vdots \\
\gamma_{p1} & \dots & \gamma_{pK}
\end{matrix}\bigg]
\end{align}
has full column rank. In the following, we use $\bgamma_j$, $j = 1,...,p$, to denote its rows.

We focus on settings with many potentially correlated risk factors, where $K$ is large but fixed and the matrix $\Pi$ may suffer from multicollinearity. In such cases, standard MVMR estimators can be unstable and biased, particularly when estimation error in $\hat\Pi$ is large.

\subsection{Two-sample summary-data MVMR} \label{sec: setup}

Our framework is based on the two-sample summary-data MVMR setting, in which SNP–exposure and SNP–outcome associations are obtained from non-overlapping GWASs, while allowing for sample overlap across the exposure datasets. For each SNP $j=1,\dots,p$, we observe the vector of estimated SNP–exposure associations $\hat \bgamma_j = (\hat \gamma_{j1},\dots,\hat \gamma_{jK})^T$ from marginal regressions of each standardized exposure on SNP $j$, along with their standard errors $(\sigma_{Xj1},\dots,\sigma_{XjK})^T$. We also observe the SNP–outcome association estimate $\hat \Gamma_j$ from the marginal regression of the outcome $Y$ on SNP $j$, together with its standard error $\sigma_{Yj}$.

To account for potential sample overlap between exposure GWASs, we incorporate a shared positive-definite correlation matrix $\Sigma$, which captures correlations among elements of $\hat \bgamma_j$. The variance-covariance matrix of $\hat \bgamma_j$ is
$$
\Sigma_{Xj} = {\rm diag}(\sigma_{Xj1},...,\sigma_{XjK}) \Sigma {\rm diag}(\sigma_{Xj1},...,\sigma_{XjK}).
$$
Although $\Sigma$ is not directly available from GWAS summary statistics, it can be estimated using bivariate LD score regression \citep{bulik2015atlas} or from empirical correlations between GWAS $Z$-scores across null SNPs \citep{Wang2021grapple}.

{Given the large sample sizes in modern genome-wide association studies, we assume for each SNP $j$, $\hat \bgamma_j \sim N(\bgamma_j, \Sigma_{Xj})$ and $\hat \Gamma_j \sim N(\Gamma_j, \sigma_{Yj}^2)$ (see Assumption 1 in Supplement).}

\subsection{Instrument Selection and Data Preprocessing} \label{sec: iv selection}

We followed the approach of \cite{Wang2021grapple} to select $p$ SNPs associated with at least one risk factor at a genome-wide significance threshold of $5\times10^{-8}$. After initial screening, LD clumping was applied to obtain a set of approximately independent SNPs, a procedure demonstrated in our data application. When risk factors are measured on different scales, we recommend rescaling SNP–exposure associations and their standard errors by the reported standard deviation of each risk factor, ensuring that all estimated causal effects are comparable.

Throughout, we assume that the number of SNPs exceeds the number of risk factors and that these SNPs satisfy the instrumental variable assumptions. In practice, additional filtering methods \citep{Verbanck:2018aa,chan2024novel} can be incorporated during SNP selection to help exclude potentially invalid instruments and reduce the risk of violating the exclusion restriction.

\subsection{MVMR-PACS} \label{sec: mvmr pacs}

We propose MVMR-PACS, a novel estimator for MVMR designed for settings with many highly correlated risk factors. Unlike existing methods, MVMR-PACS explicitly accounts for the fact that $\hat \Pi$ is a random matrix with entries measured with error.

Most existing MVMR estimators ignore this source of uncertainty. For example, the widely used inverse-variance weighted (IVW) estimator \citep{burgess2015multivariable} minimizes
\begin{align}
    \mathcal{L}_{\rm IVW}(\bbeta) & = \frac{1}{2}(\hat \bGamma - \hat \Pi \bbeta)^T W (\hat \bGamma - \hat \Pi \bbeta) \label{eq: IVW} \nonumber \\
    & = \frac{1}{2} \bbeta \hat  \Pi^T W \hat \Pi \bbeta - \hat \bGamma^T W \hat \Pi \bbeta  + c,
\end{align}
where $\hat \bGamma = (\hat \Gamma_1,...,\hat \Gamma_p)$, $W = {\rm diag}(1/\sigma_{Y1}^2,...,1/\sigma_{Yp}^2)$ and $c$ is a constant independent of $\bbeta$.  The IVW estimator is known to suffer bias when instruments are weak \citep{wu2024more}. This bias arises because the loss function it minimizes differs from the oracle form,
\begin{align}
    \mathcal{L}_{\rm IVW}^{\rm oracle}(\bbeta) = \frac{1}{2} \bbeta  \Pi^T W \Pi \bbeta - \hat \bGamma^T W \Pi \bbeta  + c\label{eq: IVW true},
\end{align}
with the discrepancy between $\Pi$ and its estimate $\hat \Pi$ driving the bias when estimation error in $\hat \Pi$ is non-negligible.

To address this, we debias the IVW loss by noting that
\begin{align}
  &  E(\hat  \Pi^T W \hat \Pi) = \Pi^T W \Pi + V, \nonumber
  \quad E(\hat \bGamma^T W \hat \Pi) = \bGamma^T W \Pi \nonumber,
\end{align}
where $V = \sum_{j = 1}^p \Sigma_{Xj}\sigma_{Yj}^{-2}$. Minimizing the debiased loss function,
\begin{align}
\mathcal{L}_{\rm debiased}(\bbeta) = \tfrac{1}{2}\bbeta^T (\hat \Pi^T W \hat \Pi - V)\bbeta - \hat \bGamma^T W \hat \Pi \bbeta, \label{eq: dIVW}
\end{align}
is, on average, equivalent to minimizing the oracle loss in \eqref{eq: IVW true}. When $(\hat \Pi^T W \hat \Pi - V)$ is positive definite, this yields the debiased IVW estimator \citep{wu2024more}.

A challenge arises when $(\hat \Pi^T W \hat \Pi - V)$ is not positive definite, in which case the loss is non-convex and unbounded below. To restore convexity, we adopt the projection technique of \cite{datta2017cocolasso} and define the projected debiased loss function:
\begin{align}
    \mathcal{L}_{\rm debiased}^+(\bbeta) = \frac{1}{2} \bbeta^T  (\hat \Pi^T W \hat \Pi - V)_{+} \bbeta - \hat \bGamma^T W \hat \Pi \bbeta  \label{eq: dIVW +},
\end{align}
where $(M)_{+}$ denotes the nearest positive semi-definite matrix to $M$, computed as
\begin{align}
    (M)_{+} = \argmin_{M^{'} > 0} ||M-M^{'}||_{\rm max},
\end{align}
where $M^{'}$ is any positive definite matrix and $||M||_{\rm max} = \max_{i,j}|M_{ij}|$ denotes the elementwise maximum norm. 
This projection ensures convexity and stability by discarding information along directions in which causal effects are effectively unidentified (e.g., contrasts between almost collinear exposures), while preserving information in directions with sufficient instrument strength.

Building on this projected debiased loss, we define the MVMR-PACS estimator $\hat\bbeta_{\rm PACS}$ as the minimizer of
\begin{align}
    \mathcal{L}_{\rm PACS}(\bbeta, \lambda) & = \frac{1}{2} \bbeta^T  (\hat \Pi^T W \hat \Pi - V)_{+} \bbeta - \hat \bGamma^T W \hat \Pi \bbeta  \nonumber \\
    & + \lambda\bigg\{\sum_{k = 1}^K w_k |\beta_k| + \sum_{k < m} w_{km(-)} |\beta_k - \beta_m| + \sum_{k < m} w_{km(+)}|\beta_k + \beta_m|\bigg\} \label{eq: mvmr pacs},
\end{align}
where $\lambda$ is a tuning parameter and $w_k$, $w_{km(-)}$, and $w_{km(+)}$ are non-negative weights. Alternative forms of penalization are reviewed in Supplementary Section 3. The PACS penalty is particularly attractive because it adapts to the correlation structure of the exposures and does not require prior specification of grouping.

We adopt adaptive weighting schemes following \cite{sharma2013consistent}. The weights depend on an initial consistent estimator and the observed correlations $\hat r_{km}$ between SNP–exposure associations for exposures $k$ and $m$. The weights are defined as 
$$w_{k} = |\tilde{\beta_k}|^{-\tau}, \quad w_{km(-)} = (1-\hat r_{km})^{-\tau}|\tilde \beta_k - \tilde \beta_m|^{-\tau}, \quad w_{km(+)} = (1+\hat r_{km})^{-\tau}|\tilde \beta_k + \tilde \beta_m|^{-\tau}$$ 
for $k, m = 1,\dots,K$,
where $\tilde \bbeta_{\rm init} = (\tilde \beta_1,\dots,\tilde \beta_K)^T$ is a consistent and asymptotically normal initial estimator of $\bbeta$ under the conditions of Theorem \ref{theorem: PACS}, and $\tau$ is another tuning parameter. This construction ensures weaker penalization for effects or pairwise sums/contrasts that are large in magnitude, a critical feature for achieving the oracle property \citep{zou2006adaptive,sharma2013consistent}. Tuning parameters are selected using multi-fold data thinning (see Sections \ref{sec: data thin}).

By incorporating correlations between SNP–exposure associations, MVMR-PACS shrinks coefficients toward each other more heavily when $\hat r_{km}$ is close to one, or toward opposite values when $\hat r_{km}$ is close to minus one. This improves stability in the presence of highly correlated exposures that cannot be disentangled due to weak conditional instruments. When associations are not highly correlated, the method estimates each risk factor’s effect individually. 

We use the minimizer of $\mathcal{L}_{\mathrm{debiased}}^{+}(\boldsymbol{\beta})$ with a ridge penalty as the initial estimator (its consistency and asymptotic normality are established in Supplementary Section 8) and then  minimize $\mathcal{L}_{\rm PACS}(\bbeta,\lambda)$ using the local quadratic approximation algorithm of \citet{sharma2013consistent}, which provides efficient closed-form updates (see Supplementary Section 7 for details). When $w_{km(-)}=w_{km(+)}=0$, the PACS penalty reduces to the adaptive LASSO penalty \citep{zou2006adaptive}, yielding the MVMR-dLASSO estimator. As illustrated by MVMR-PACS-0.8 in our simulation study, we also consider a family of variants, denoted MVMR-PACS-$x$, which incorporate a correlation threshold of value $x$ into the weighting scheme to control the degree of grouping among correlated exposures (see Supplementary Section 5 for details).

\subsection{Signal-groups and estimand from MVMR-PACS} \label{sec: interpretation}

MVMR-PACS produces clusters of exposures, which we refer to as \emph{signal-groups}. A signal-group is defined as a set of risk factors that share nonzero effect estimates of equal magnitude (up to a specified decimal precision). Each signal-group can be viewed as a linear combination—typically a sum or difference—of the risk factors it contains, with the direction determined by the signs of their estimated effects. The key question is how to interpret the shared effect for each signal-group.

To illustrate, consider $K_1+K_2$ exposures. Suppose MVMR-PACS identifies two signal-groups: $G_1$, comprising the first $K_1$ exposures, and $G_2$, comprising the remaining $K_2$. A working model then estimates the direct causal effects of $G_1$ and $G_2$. As shown in Supplementary Section 4, these effects are linear combinations of the true causal effects of all exposures, including those outside the respective groups. The grouping becomes interpretable when the exposures within a signal-group are strongly correlated in their SNP–exposure associations or genuinely share the same causal effect, in which case the estimated group effect approximates a weighted average of the individual effects with positive weights.

When multicollinearity is severe, the true model may be unidentifiable, making the working model not only a pragmatic approximation but also a necessary alternative for stable estimation. Indeed, existing high-dimensional MVMR methods also rely on working models for effect estimation and variable selection. For example, MR-BMA averages over many combinations of risk factors, and MVMR-cML-SuSiE selects representative exposures from causal clusters. Each combination implicitly defines a working model. However, these approaches provide limited discussion of how to interpret the resulting estimands or of the conditions under which such working models approximate the true model.


\subsection{Data thinning for summary-data MVMR} \label{sec: data thin}

We introduce a data-thinning approach \citep{neufeld2023data} to generate independent replicates of GWAS summary statistics, which can then be used for tasks such as cross-validation and post-selection inference.

In the MVMR literature, cross-validation is often implemented by sample splitting, where SNPs are randomly partitioned into folds \citep{grant2022efficient, zhao2023robust}. This approach is often infeasible when the number of instruments is small and SNP-level association distributions are heterogeneous. For post-selection inference, existing approaches generally use the same dataset for both model selection and inference, which can result in underestimated standard errors and inflated type I error due to “double dipping.’’ Although this limitation has been noted in prior work \citep{grant2022efficient, chan2024novel}, it remains unresolved. 

Data thinning provides a way forward \citep{neufeld2023data}. The idea is to decompose an observed SNP–exposure or SNP–outcome association into two or more independent parts that sum to the original estimate while retaining the same distributional structure, up to a known scaling factor. {For example, in two-fold thinning, for each SNP $j$ one can create two independent SNP-exposure association estimates by
\begin{align*}
    \hat \bgamma_j^{(1)} = \frac{\hat \bgamma_j}{2} + \epsilon_j, \qquad  \hat \bgamma_j^{(2)} = \hat \bgamma_j - \hat \bgamma_j^{(1)}, 
\end{align*}
where $\epsilon_j$ is drawn independently from $N(0,\frac{\Sigma_{Xj}}{4})$, 
so that marginally,
\begin{align*}
     \hat \bgamma_j^{(1)} \sim N(\frac{\bgamma_j}{2},\frac{\Sigma_{Xj}}{2}), \qquad  \hat \bgamma_j^{(2)} \sim N(\frac{\bgamma_j}{2},\frac{\Sigma_{Xj}}{2}),
\end{align*}
with $\hat \bgamma_j^{(1)}$ and $\hat \bgamma_j^{(2)}$ independent. SNP–outcome associations can be decomposed analogously.} The resulting two independent datasets, $\left \{\hat \bgamma_j^{(1)},\Sigma_{Xj}/2,\hat \Gamma_j^{(1)},\sigma_{Yj}^2/2, j = 1,...,p \right \}$ and\\ $\left \{\hat \bgamma_j^{(2)},\Sigma_{Xj}/2,\hat \Gamma_j^{(2)},\sigma_{Yj}^2/2, j = 1,...,p \right \}$, each contain all $p$ SNPs and preserve the causal relationship in equation \eqref{eq: true causal}, while providing independent replicates for downstream analysis. The main consequence of generating independent replicates is a reduction in instrument strength. As shown in Supplementary Section 6, the instrument strength parameter \cite{wu2024more} is halved under standard two-fold thinning, making it necessary to use weak-instrument robust methods for inference.

Although the above procedure partitions data into two parts, it can be extended to create $M$ independent datasets with potentially uneven information allocation, a strategy we refer to as multi-fold data thinning. A general procedure for multi-fold thinning in MVMR is provided in Supplementary Section 6.

\subsection{Cross-validation} \label{sec: tuning paramter}

We use five-fold data thinning and select tuning parameters by minimizing the projected debiased loss function
\begin{align}
    \sum_{i=1}^5 \frac{1}{2} \hat \bbeta_{-i}(\bm\delta)^T  (\hat \Pi_{i}^T W_i \hat \Pi_{i} - V_{i})_{+} \hat \bbeta_{-i}(\bm\delta) - \hat \bGamma_{i}^T W_i \hat \Pi_i \hat \bbeta_{-i}(\bm\delta)   \label{obj: cv} ,
\end{align}
where $\bm \delta = (\lambda, \tau)$ denotes the vector of tuning parameters, $\hat \Pi_i$, $W_i$, $V_i$, $\hat \bGamma_i$ are based on the $i$th replicate of the summary statistics, and $\hat \bbeta_{-i}(\bm\delta)$ is obtained based on the summary statistics summed over the remaining replicates (see Supplementary Section 6).

In high-dimensional settings, parsimony is often desirable. To encourage simpler models, we apply the one-standard-error (1SE) rule. After computing the cross-validated loss for a grid of candidate tuning parameter values, we identify the parameter choice that minimizes the average loss and calculate the corresponding standard error. {Among all candidates whose average loss lies within one standard error of the minimum, we first select the largest $\lambda$, then the largest $\tau$ among candidates sharing that $\lambda$ value.} This strategy promotes sparsity and grouping of correlated risk factors while maintaining near-optimal predictive accuracy, thereby mitigating overfitting and reducing false discoveries. Choices of $\lambda$ and $\tau$ are further discussed in the next section, where we present the theoretical results. 

\subsection{Theoretical analysis} \label{sec: theory}

\subsubsection{MVMR-PACS}

We establish that MVMR-PACS enjoys the oracle properties of variable selection and signal-group identification. Let $\bbeta^* = (\beta_1^*, \beta_2^*,\dots, \beta_{K_1}^*,\beta_{K_1+1}^*,\dots,\beta_{K_1+K_0}^*)^T$ denote the vector of true causal effects for $K = K_1 + K_0$ risk factors,
where the first $K_1$ are causal ($\beta_k^* \neq 0$ for $k=1,\dots,K_1$) and the remaining $K_0$ are non-causal. Suppose the causal risk factors can be partitioned into signal-groups $\mathcal{G} = \{G_1,\dots,G_L\}$ such that within each group $G_l$ all effects share the same magnitude, $|\beta_k^*| = \alpha_l$ for $k \in G_l$, with $\alpha_l > 0$. Since signs may differ within a group, we write $\beta_k^* = s_k \alpha_l$ for $k \in G_l$, where $s_k \in \{-1,+1\}$ denotes the sign.

To uniquely represent grouped effects, we define a canonical representative for each group. For $l=1,\dots,L$, let $k_l = \min\{k: k \in G_l\}$ denote the smallest index in group $G_l$, and set the group sign as $\bar s_l = s_{k_l} \in \{-1,+1\}$. We then construct an $L \times K$ matrix $C_g$ with entries
\begin{align*}
    C_{g, lk} = \begin{cases}
        \frac{\bar s_l s_{k}}{|G_l|}, & \text{$k \in G_l$} \\
        0, & \text{otherwise}
    \end{cases},
\end{align*}
where $|G_l|$ is the size of group $G_l$. The grouped causal effects are defined as $\bar \bbeta^* = (\bar \beta_1^*,\dots,\bar \beta_L^*)^T =  C_g \bbeta^* \in \mathbb{R}^L$, with
\begin{align*}
    \bar \beta_l^* = \frac{\bar s_l}{|G_l|}\sum_{k\in G_l} s_k\beta_k^*.
\end{align*}
As an illustration, consider five risk factors with $\beta_1^* = \beta_2^*$, $\beta_3^* = -\beta_4^*$, and $\beta_5^*=0$, where $|\beta_1| \neq |\beta_3|$. This yields three groups (two causal and one non-causal) and the grouping matrix 
$$C_g = \begin{bmatrix}
    0.5 & 0.5 & 0 & 0 & 0 \\
    0 & 0 & 0.5 & -0.5 & 0
\end{bmatrix},$$ 
which maps $\bbeta^*$ to the group-level effects $\bar \bbeta^*$ and drops the exposures with no effects.

To study the consistency of variable selection and group identification, we follow \cite{sharma2013consistent} and define the over-parameterized vector $\bm \theta^* = M \bbeta^*$, where $M$ is a $K^2 \times K$ matrix given by $M = [I_K \ D_{(-)}^T \ D_{(+)}^T]^T$. The matrix $D_{(-)}$ has dimension $\frac{K(K-1)}{2} \times K$ with entries $\pm 1$, such that $D_{(-)}\bbeta$ yields all pairwise differences of $\bbeta$ (for example, $\beta_2 - \beta_1$, $\beta_3 - \beta_1$, \dots, $\beta_K - \beta_{K-1}$). Similarly, $D_{(+)}$ is a $\frac{K(K-1)}{2} \times K$ matrix with entries equal to $+ 1$, such that $D_{(+)}\bbeta$ gives all pairwise sums of $\bbeta$. 

Let $\mathcal{B} = \{k: \theta_k^* \neq 0, k = 1,...,K^2\}$ denote the set of nonzero indices of $\bm \theta^*$. Its complement, $\mathcal{B}^c$, specifies which risk factors have no causal effects and which share effects of identical magnitude and direction, thereby defining the oracle sparsity pattern and grouping structure. Let $\mathcal{B}_n = \{k: \hat \theta_k^* \neq 0, k=1,\dots,K^2\}$ denote the corresponding set estimated by MVMR-PACS, where $\hat {\bm \theta}^* = M \hat \bbeta_{\rm PACS}$ and $\hat \bbeta_{\rm PACS}$ is defined in equation \eqref{eq: mvmr pacs}.

Let $\bar \bgamma_j = G \bgamma_j$ denote the true association of SNP $j$ with the $L$ causal groups, where $G$ is an $L\times K$ matrix analogous to $C_g$ but without the $|G_l|$ scaling. Specifically,
\begin{align*}
    G_{lk} = \begin{cases}
        \bar s_l s_{k}, & \text{$k \in G_l$} \\
        0, & \text{otherwise}
    \end{cases},
\end{align*}
so that $G$ collapses elements of $\bgamma_j$: it sums those with nonzero effects of the same magnitude and direction, subtracts those with opposite effects, and omits those with zero effects.

\begin{theorem} \label{theorem: PACS}
    (Oracle properties of MVMR-PACS)  Under the standard MVMR Assumptions 1-2 (see Supplementary Section 8), if ${\mu_{n,\min}}/{\sqrt{p}}\rightarrow \infty$ and $\max_j \gamma_{jk}^2\sigma_{Xjk}^{-2}/(\mu_{n,\min} + p) \rightarrow 0$ for all $k$ as $n\to\infty$, and if the tuning parameters satisfy $\lambda_n/r_n \rightarrow 0$, $\lambda_n r_n^{\tau - 1} \rightarrow \infty$ with $r_n = \mu_{n,\min}/\sqrt{\mu_{n,\min} + p}$ and $\lambda_nr_n^{\tau+1}/\mu_{n,\max}\rightarrow \infty$, then the MVMR-PACS estimator has the following properties:
    \begin{enumerate}
        \item Consistency in risk factor selection: $\lim_n \probP(\mathcal{B}_n = \mathcal{B}) = 1$. 
        \item Asymptotic normality: $
    \mathbb{V}_{g}^{-\frac{1}{2}} (\Pi_g^T W \Pi_g)(C_g\hat \bbeta_{{\rm PACS}} - C_g\bbeta^*)  \xrightarrow[]{D} N(\bm 0, I_{L})$.
    \end{enumerate}
Here, $\mathcal{B} = \{k: \theta_k \neq 0, k=1,...,K^2\}$, $\mathcal{B}_n = \{k:  \hat \theta_k^* \neq 0\}$, $\Pi_g =  \Pi G^T$, $V_{j, g} = G \Sigma_{Xj}\sigma_{Yj}^{-2}G^T$, $\mathbb{V}_g = \sum_{j=1}^{p}\big\{  (1+ \bar \bbeta^{*T}V_{j,g}\bar\bbeta^*) (G\bgamma_j\bgamma_j^T\sigma_{Yj}^{-2}G^T+ V_{j,g})+V_{j,g}\bar \bbeta^*\bar \bbeta^{*T}V_{j,g}\big\}$, with $I_{L}$ denoting the $L\times L$ identity matrix.
\end{theorem}

The proof is provided in the Supplement. Assumptions 1 and 2 are standard and were also used in \cite{wu2024more}. The scalar $\mu_{n,\min}$ is an instrument strength parameter that characterizes the rate at which instrument strength increases along the weakest direction among all linear combinations of exposures as $p\to \infty$, which determines the asymptotic behavior of MVMR estimators. The condition $\max_j (\gamma_{jk}^2{\sigma_{Xjk}^{-2}} ) /(\mu_{n,\min} + p) \rightarrow 0$ for all $k$ ensures that no single SNP dominates the instrument strength and is required to verify Lindeberg’s condition for the central limit theorem. The requirement ${\mu_{n,\min}}/{\sqrt{p}} \rightarrow \infty$ sets a minimum instrument strength condition needed for consistency and asymptotic normality of the initial estimator. By contrast, \citet{wu2024more} showed that MVMR-IVW requires the stronger condition $\mu_{n,\min}/p^2 \to \infty$ to achieve consistency, which is generally impractical. Thus, ${\mu_{n,\min}}/{\sqrt{p}} \to \infty$ represents a substantially weaker and more realistic condition on instrument strength.

The additional requirements $\lambda_n/r_n \to 0$, $\lambda_n r_n^{\tau - 1} \to \infty$, and $\lambda_n r_n^{\tau + 1}/\mu_{n,\max}\to \infty$ are conditions on the tuning parameters.
One regime where these conditions are satisfied is when $\mu_{n,\min} = \Theta(p)$, a typical rate under many weak instruments. We write $a = \Theta(b)$ if there exists a constant $c > 0$ such that $c^{-1} b \leq |a| \leq c b$. In this setting, $r_n = \Theta(\sqrt{p})$ and $\lambda_n = o(\sqrt{p})$. The maximum strength parameter $\mu_{n,\max} $ typically grows no faster than $\Theta(n)$, and in most MVMR studies it is reasonable to assume $n = O(p^2)$. Under this regime, the conditions $\lambda_n r_n^{\tau-1} \to \infty$ and $\lambda_n r_n^{\tau+1}/\mu_{n,\max} \to \infty$ are readily achievable. For instance, setting $\tau = 3$ and $\lambda_n = O(p^{1/4})$ suffices.

\subsubsection{Choices of tuning parameters}

In practice, $\mu_{n,\min}$ is estimated by $\hat \mu_{n,\min}$, the minimum eigenvalue of the sample instrument strength matrix \citep{wu2024more}. We performed 5-fold data thinning to choose $\lambda_n$ from a grid of values of the rate $(\hat \mu_{n,\min}/\sqrt{\hat \mu_{n,\min} + p})^{2/3}$ if $\hat \mu_{n,\min} > p$, or $(p/2)^{1/3}$ otherwise, and chose $\tau$ from $(0.5,1,2,3)$.

\subsection{Simulation studies} \label{sec: sim details}

We generated summary statistics by first simulating individual-level data. We considered $K=10$ risk factors and $p = 500$ SNPs, with the correlation structure of SNP–exposure associations mimicking that in Figure~1c. Specifically, 20\% of SNPs were associated with all risk factors, 40\% with the first six, and the remaining 40\% with the last four. True $\gamma_{jk}$ values were drawn from mean-zero normal distributions and fixed across simulations.

Let $\Sigma_{\rho, K}$ denote a $K\times K$ compound symmetry matrix with diagonal entries equal to 1 and off-diagonal entries equal to $\rho$, where $\rho$ controls the level of within-cluster correlations. Let $\bm c_{a\times b}$ be an $a\times b$ matrix of constant $c$. For SNPs associated with the first six risk factors,
\begin{align*}
    (\gamma_{j1},\dots,\gamma_{j6})^T \sim N(\bm 0, \sigma_{\gamma}^2 \Sigma_{\rm cluster 1}) \quad \text{and} \quad (\gamma_{7},\dots,\gamma_{j10})^T = \bm 0,
\end{align*} 
where 
\begin{align*}
    \Sigma_{\rm cluster 1} = \begin{bmatrix}
        \Sigma_{0.995, 3} & {\bm {0.5}}_{3\times 3} \\
        {\bm {0.5}}_{3\times 3} & \Sigma_{0.9, 3}
    \end{bmatrix}.
\end{align*}
For SNPs associated with the last four risk factors,
\begin{align*}
    (\gamma_{j1},\dots,\gamma_{j6})^T = \bm 0 \quad \text{and} \quad (\gamma_{j7},\dots,\gamma_{j10})^T \sim N(\bm 0, \sigma_{\gamma}^2 \Sigma_{\rm cluster 2}),
\end{align*}
with $\Sigma_{\rm cluster 2} = \Sigma_{0.3, 4}$. For SNPs associated with all ten risk factors,
\begin{align*}
    (\gamma_{j1},\dots,\gamma_{j10})^T \sim N(\bm 0, \sigma_{\gamma}^2 \Sigma_{\rm all}),
\end{align*}
where 
\begin{align*}
    \Sigma_{\rm all} = \begin{bmatrix}
        \Sigma_{\rm cluster 1} & {\bm {0.3}}_{4\times 4} \\
        {\bm {0.3}}_{4\times 4} & \Sigma_{\rm cluster 2}
    \end{bmatrix}.
\end{align*}

For each individual $i = 1,...,n$, SNPs $Z_{ij}$ were generated with minor allele frequencies $\text{MAF}_j \sim \text{Unif}(0.01,0.5)$. An unmeasured confounder $U_i \sim N(0, \sigma_u^2)$ was generated. The exposures $X_{ik}$'s and outcome $Y_{i}$ were simulated as
\begin{align*}
    & X_{ik} = \bm Z_i^T \bgamma_{k} + U_i + e_{ik},\quad k = 1,\dots,K, \\
    & Y_i = \bm X_i^T \bbeta^* + U_i + E_i
\end{align*}
where $\bm Z_i = (Z_{i1},\dots,Z_{ip})^T$, $\bm X_i = s(X_{i1},\dots,X_{iK})^T$ representing the standardized vector of exposures, $e_{ik} \sim N(0, \sigma_{e}^2)$, and $E_i \sim N(0, \sigma_{e}^2)$.  We set $\sigma_{\gamma} = 0.001$, $\sigma_{u} = 2$ and $\sigma_{e} = 1$, corresponding to a many–weak–instruments setting with approximately 2\% heritability across the risk factors. 

This process was repeated to generate two independent datasets. Marginal SNP–exposure associations and their standard errors were obtained from simple linear regressions in one dataset, and SNP–outcome associations and their standard errors from regressions in the other. The true causal effects were set to $\bbeta^* = (1,1,1,0,0,0,0,0,0.5,0)$, with sample sizes $n = 1\times 10^5,\ 2\times 10^5,\ \text{and}\ 3\times 10^5$.

We compared MVMR-PACS against alternative estimators listed in Table \ref{tab:methods_summary}. IVW-LASSO was tuned by the 1SE rule. MR-BMA was implemented with a prior inclusion probability of 0.5 and default hyperparameters. For MVMR-cML-SuSiE, we first screened the ten risk factors using univariable MR-cML \citep{xue2021constrained}, then ran MVMR-cML-SuSiE initialized with the MVMR-cML \citep{lin2023robust} estimates for the selected factors. Risk factor selection was assessed by inclusion in any causal cluster (posterior inclusion probability greater than $1/K_s$, where $K_s$ is the number of selected risk factors from the initial screening).

\section*{Data Availability}

All data used in our study is in the public domain. Our study is based on publicly available summary-level data on genetic associations, which are referenced in Supplementary Table S5.

\section*{Code Availability}

R code for MVMR-PACS is available at \url{https://github.com/yinxiangwu/MVMR-PACS}; R code for MR-BMA and IVW-Lasso is avaialble at \url{https://github.com/verena-zuber/demo_AMD}; R code for MVMR-cML-SuSiE is available at \url{https://github.com/lapsumchan/MVMR-cML-SuSiE}; R code for SRIVW is available at \url{https://github.com/tye27/mr.divw}. R code for the simulation study and real data application is available at \url{https://github.com/yinxiangwu/MVMR-PACS}.

\clearpage

\spacingset{1.0}
\bibliography{reference}

\section*{Acknowledgments}
T.Y.’s research was partially supported by the National Institute Of General
Medical Sciences of the National Institutes of Health under Award Number R35GM155070.

\section*{Author Contributions}

Y.W. and T.Y. conceived and designed the study. Y.W. performed the statistical analyses and theoretical development under the supervision of T.Y. N.M.D. contributed to the study design and interpretation of the results. Y.W. and T.Y. wrote the manuscript, and all authors reviewed and approved this version.

\subsection*{Corresponding Author}
Correspondence to Ting Ye.

\clearpage

\pagenumbering{arabic}
\setcounter{equation}{0}
\setcounter{table}{0}
\setcounter{figure}{0}
\setcounter{section}{0}
\setcounter{lemma}{0}
\setcounter{assumption}{0}
\setcounter{theorem}{0}
\renewcommand{\theequation}{S\arabic{equation}}
\renewcommand{\thetable}{S\arabic{table}}
\renewcommand{\thelemma}{S\arabic{lemma}}
\renewcommand{\thesection}{S\arabic{section}}
\renewcommand{\thefigure}{S\arabic{figure}}

\title{\LARGE Supplement to
``Group Identification and Variable Selection in Multivariable Mendelian Randomization
with Highly-Correlated Exposures''}
\maketitle

\section{Additional simulation details and results}

\begin{table}[ht]
\centering
\caption{Numerical results from simulations with sample sizes $N=1\times10^5$, $2\times10^5$, and $3\times10^5$, reporting median mean squared error (MSE), correct sparsity, sensitivity, and false positive rate across 1000 replicates. MVMR-PACS-0.8 is a variant of MVMR-PACS, incorporating a correlation threshold of 0.8 into its weighting scheme. MR-BMA-0.2, MR-BMA-0.5, and MR-BMA-0.8 denote Bayesian model averaging with MIP thresholds of 0.2, 0.5, and 0.8; they all share the same model-averaged effect estimates and hence the same MSE estimates. MSE is not reported for MVMR-cML-SuSiE and marked as ``NA''.}\label{main sim table}
\centering
\resizebox{\ifdim\width>\linewidth\linewidth\else\width\fi}{!}{
\begin{tabular}[t]{c|ccc|ccc|ccc|ccc}
\toprule
\multicolumn{1}{c}{ } & \multicolumn{3}{c}{MSE} & \multicolumn{3}{c}{Correct Sparsity} & \multicolumn{3}{c}{Sensitivity} & \multicolumn{3}{c}{False Positive} \\
\cmidrule(l{3pt}r{3pt}){2-4} \cmidrule(l{3pt}r{3pt}){5-7} \cmidrule(l{3pt}r{3pt}){8-10} \cmidrule(l{3pt}r{3pt}){11-13}
Estimator & $1\times10^5$ & $2\times10^5$ & $3\times10^5$ & $1\times10^5$ & $2\times10^5$ & $3\times10^5$ & $1\times10^5$ & $2\times10^5$ & $3\times10^5$ & $1\times10^5$ & $2\times10^5$ & $3\times10^5$\\
\midrule
MVMR-IVW & 1.087 & 0.756 & 0.622 & 0.659 & 0.737 & 0.763 & 0.225 & 0.396 & 0.448 & 0.052 & 0.035 & 0.026\\
\cmidrule{1-13}
SRIVW & 7.047 & 9.585 & 6.905 & 0.736 & 0.708 & 0.701 & 0.342 & 0.274 & 0.258 & 0.002 & 0.003 & 0.003\\
\cmidrule{1-13}
IVW-LASSO & 1.146 & 0.951 & 0.825 & 0.849 & 0.895 & 0.946 & 0.699 & 0.783 & 0.893 & 0.051 & 0.031 & 0.018\\
\cmidrule{1-13}
MVMR-dLASSO & 5.757 & 5.759 & 4.387 & 0.73 & 0.771 & 0.821 & 0.324 & 0.427 & 0.552 & 0.000 & 0.000 & 0.000\\
\cmidrule{1-13}
MVMR-PACS & 0.28 & 0.215 & 0.102 & 0.883 & 0.944 & 0.97 & 0.835 & 0.893 & 0.945 & 0.085 & 0.022 & 0.013\\
\cmidrule{1-13}
MVMR-PACS-0.8 & 0.255 & 0.25 & 0.071 & 0.909 & 0.941 & 0.976 & 0.777 & 0.855 & 0.942 & 0.003 & 0.002 & 0.001\\
\cmidrule{1-13}
MRBMA-0.2 & 0.57 & 0.369 & 0.279 & 0.413 & 0.498 & 0.579 & 0.984 & 0.997 & 1.000 & 0.968 & 0.835 & 0.702\\
\cmidrule{1-13}
MRBMA-0.5 & 0.57 & 0.369 & 0.279 & 0.753 & 0.853 & 0.896 & 0.862 & 0.973 & 0.994 & 0.319 & 0.228 & 0.169\\
\cmidrule{1-13}
MRBMA-0.8 & 0.57 & 0.369 & 0.279 & 0.817 & 0.898 & 0.938 & 0.691 & 0.851 & 0.924 & 0.098 & 0.071 & 0.053\\
\cmidrule{1-13}
MVMR-cML-SuSiE & NA & NA & NA & 0.738 & 0.757 & 0.799 & 0.451 & 0.518 & 0.578 & 0.071 & 0.083 & 0.054\\
\bottomrule
\end{tabular}}
\end{table}

\clearpage

\subsection{Graphical results}

\begin{figure}[H]
    \centering
    \includegraphics[width=10cm]{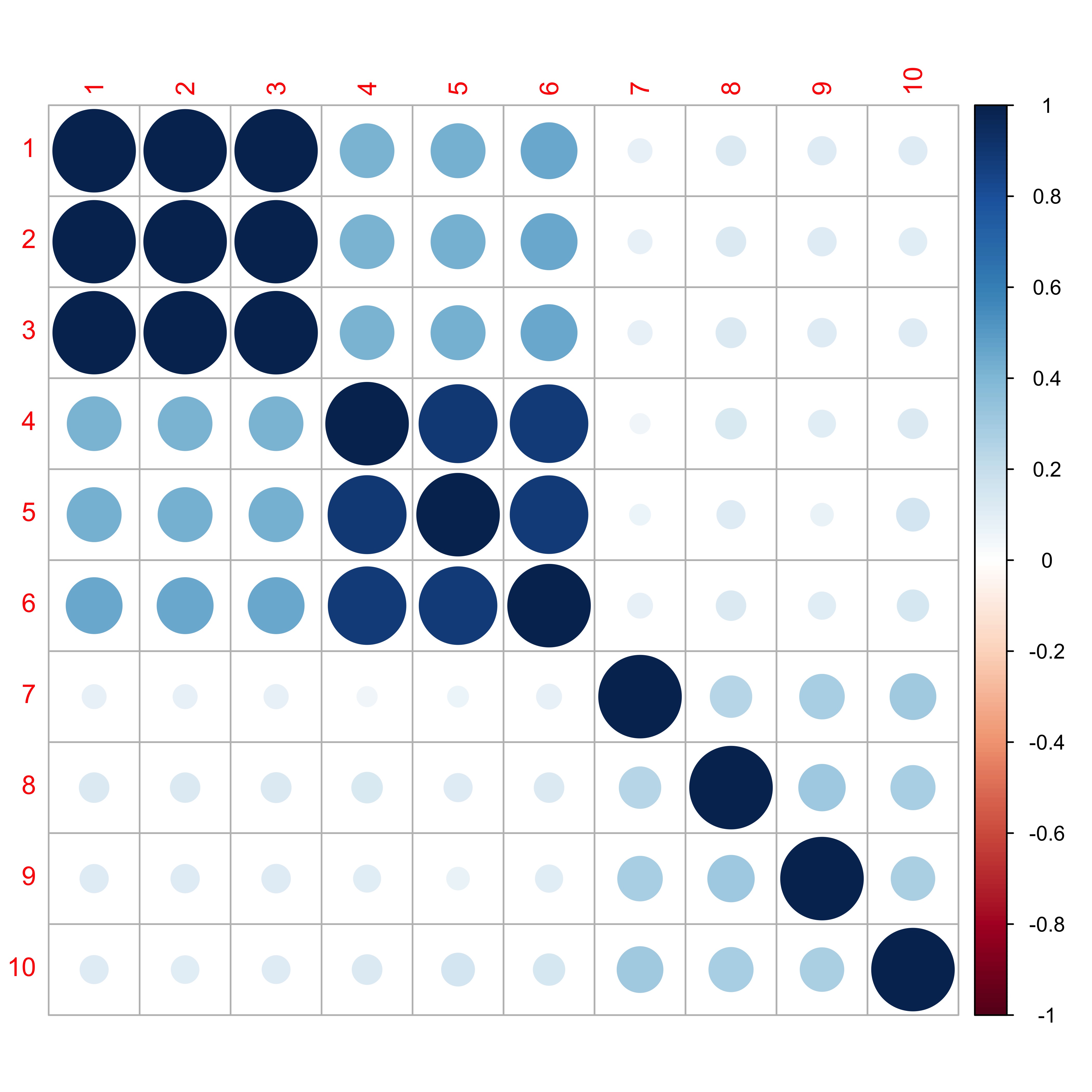}
    \caption{{\protect\footnotesize The correlation matrix of \textit{true} SNP-exposure associations across 500 SNPs.}}
\end{figure}

\begin{figure}[H]
    \centering
    \includegraphics[width=\textwidth]{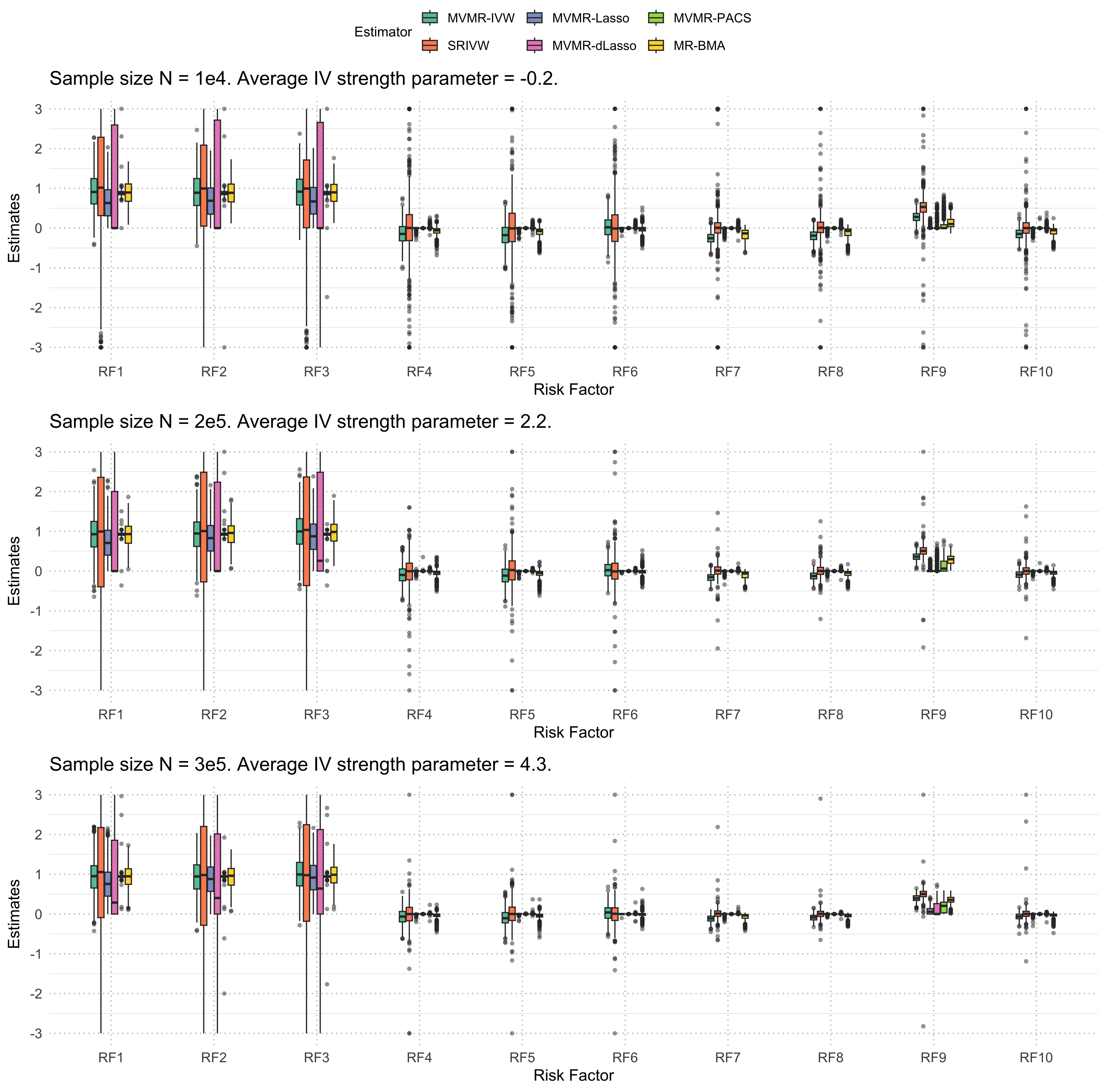}
    \caption{{\protect\footnotesize Boxplots of point estimates from different estimators across different exposures and sample sizes. Y-axis is cupped between -3 and 3 for better visibility. Abbreviations: RF: risk factor.}}
\end{figure}

\subsection{Post-selection inference results}

\begin{table}[H]
    \centering
    \begin{tabular}{c|c|c|c|c}
        \toprule
        Sample size & instrument strength & Correct Sparsity & Sensitivity & False Positive \\
        \midrule
        1e5 & 44.5 & 0.834 & 0.823 & 0.158 \\
        2e5 & 86.7 & 0.910 & 0.850 & 0.050 \\
        3e5 & 124.6 & 0.923 & 0.871 & 0.042 \\
        \bottomrule
    \end{tabular}
    \caption{MVMR-PACS performance in the selection dataset generated through data thinning. Sample size refers to the sample size of the individual-level data that are used to generate summary statistics. The instrument strength column reports the mean instrument strength parameter based on grouped exposures in the selection dataset. Data generating process is the same as in the simulation study.}
\end{table}

\begin{table}[H]
    \centering
    \begin{tabular}{c|c|c|c|c|c}
    \toprule
    Sample size & 1-1-1-0-0-0-0-0-0-0 & 1-1-1-0-0-0-0-0-2-0 & other \\
    \midrule
    1e5 & 77.5\% & 13.2\% & 9.3\% \\
    2e5 & 61.5\% & 34.2\% & 4.3\% \\
    3e5 & 55.3\% & 44.2\% & 2.5\% \\
    \bottomrule
    \end{tabular}
    \caption{Frequency of different groupings selected by MVMR-PACS in the training dataset generated through two-fold data thinning over 1000 simulation runs. The header of the table includes four typical groupings. Risk factors assigned with the same number form a signal-group, and a risk factor is labelled 0 if it is selected as non-causal. For example, 1-1-1-0-0-0-0-0-0-0 means the first three risk factors are selected as a signal-group, while the remaining risk factors are selected as non-causal. Grouping 1-1-1-0-0-0-0-0-2-0 corresponds to the true model.}
\end{table}

\begin{table}[H]
    \centering
    \begin{tabular}{c|c|c|c}
    \toprule
    Sample size & Signal-group 1 & Signal-group 2 \\
    \midrule
    1e5 & 0.949 & 0.950 \\
    2e5 & 0.950 & 0.940  \\
    3e5 & 0.939 & 0.941  \\
    \bottomrule
    \end{tabular}
    \caption{Coverage probabilities of 95\% confidence intervals constructed by SRIVW based on the inferential dataset when the model is correctly identified. Signal-group 1 is the group of risk factors 1 to 3. Signal-group 2 includes only risk factor 9.}
\end{table}

\section{Additional real data application details}

\subsection{Data sources}

{\small
\begin{table}[H]
\centering
\caption{Data sources for exposures and outcome.}
\label{supp-table:data-source}
\setlength{\tabcolsep}{4pt}
\renewcommand{\arraystretch}{1.15}
\begin{tabularx}{\textwidth}{
  >{\RaggedRight\arraybackslash}p{3.2cm}
  >{\centering\arraybackslash}p{2cm}
  >{\RaggedRight\arraybackslash}p{1.8cm}
  >{\RaggedRight\arraybackslash}p{2.5cm}
  >{\RaggedRight\arraybackslash}p{3.2cm}
  >{\RaggedRight\arraybackslash}X
}
\toprule
\textbf{Phenotype} & \textbf{Sample size} & \textbf{Ancestry} & \textbf{Source} & \textbf{Reference} & \textbf{Download link} \\
\midrule
Apolipoprotein A-I & 393{,}193 & European & UK Biobank & \cite{richardson2020evaluating} & \url{https://opengwas.io/datasets/ieu-b-107} \\
Apolipoprotein B & 439{,}214 & European & UK Biobank & \cite{richardson2020evaluating} & \url{https://opengwas.io/datasets/ieu-b-108} \\
LDL-C & 1.32\,M & European & GLGC & \cite{graham2021power} & \url{https://csg.sph.umich.edu/willer/public/glgc-lipids2021/} \\
HDL-C & 1.32\,M & European & GLGC & \cite{graham2021power} & \url{https://csg.sph.umich.edu/willer/public/glgc-lipids2021/} \\
TG & 1.32\,M & European & GLGC & \cite{graham2021power} & \url{https://csg.sph.umich.edu/willer/public/glgc-lipids2021/} \\
Lipoprotein subfraction traits & 115{,}082 & European & UK Biobank & \cite{richardson2022characterising} & \url{https://www.ebi.ac.uk/gwas/publications/35213538} \\
Coronary artery disease &  & Mostly European & 	CARDIoGRAM-plusC4D + UK Biobank & \cite{nelson2017association} & \url{https://www.ebi.ac.uk/gwas/publications/35213538} \\
\midrule
\multicolumn{6}{l}{\footnotesize Abbreviation: M, million.} \\
\bottomrule
\end{tabularx}
\end{table}\label{data sources}
}

\subsection{MR-BMA posterior inclusion probability}

\begin{table}[H]
\centering
\caption{MR-BMA results: marginal inclusion probabilities and average effects.}
\begin{tabular}{lcc}
\toprule
trait & marginal inclusion & average effect \\
\midrule
      LDL-C &               1.000 &           1.202 \\
         TG &               0.996 &           0.408 \\
       ApoB &               0.989 &          -0.497 \\
    L-LDL-P &               0.988 &          -1.330 \\
      HDL-D &               0.975 &          -1.173 \\
   L-HDL-PL &               0.967 &           1.141 \\
   L-LDL-CE &               0.896 &           0.996 \\
    S-LDL-P &               0.882 &          -0.788 \\
 XS-VLDL-PL &               0.849 &           0.544 \\
    S-HDL-L &               0.745 &          -0.499 \\
    L-LDL-L &               0.737 &           0.648 \\
    S-LDL-C &               0.673 &          -0.525 \\
    L-LDL-C &               0.652 &           0.496 \\
      ApoA1 &               0.649 &          -0.166 \\
      IDL-C &               0.647 &          -0.392 \\
      IDL-L &               0.558 &          -0.295 \\
    M-LDL-L &               0.498 &           0.236 \\
    M-HDL-P &               0.461 &          -0.214 \\
   L-LDL-PL &               0.434 &           0.134 \\
  XS-VLDL-L &               0.415 &           0.063 \\
   L-LDL-FC &               0.363 &          -0.016 \\
    M-HDL-L &               0.360 &          -0.069 \\
    S-HDL-P &               0.342 &          -0.090 \\
    S-LDL-L &               0.315 &          -0.052 \\
   M-LDL-CE &               0.314 &           0.013 \\
   M-HDL-PL &               0.313 &          -0.044 \\
   M-HDL-CE &               0.281 &           0.012 \\
     IDL-FC &               0.227 &           0.014 \\
    M-LDL-P &               0.191 &           0.018 \\
      LDL-D &               0.124 &          -0.024 \\
  XL-HDL-TG &               0.087 &          -0.009 \\
      HDL-C &               0.071 &          -0.003 \\
\bottomrule
\end{tabular}
\end{table}

\begin{table}[H]
\centering
\caption{Results from MVMR-cML analysis}
\label{tab:mvmr_cml_results}
\begin{tabular}{lrr}
\toprule
Model & Exposure &  Estimated effect \\
\midrule
  LDL-C, ApoA1, L-LDL-FC &     LDL-C &  2.084 \\
   &     ApoA1 &  0.217 \\
   &     L-LDL-FC &   -2.037 \\
  LDL-C, ApoA1, IDL-L &     LDL-C &  2.355 \\
   &    ApoA1 &  0.238 \\
   &    IDL-L &   -2.375 \\
  LDL-C, ApoA1, IDL-FC &     LDL-C &  1.592 \\
   &     ApoA1 &  0.062 \\
   &     IDL-FC &   -1.451 \\
  LDL-C, ApoA1, IDL-C &     LDL-C &  1.519 \\
   &     ApoA1 &  0.086 \\
   &     IDL-C &   -1.325 \\
  LDL-C, HDL-C, L-LDL-FC &     LDL-C &  2.759 \\
   &     HDL-C &  0.481 \\
   &     L-LDL-FC &   -3.014 \\
  LDL-C, HDL-C, IDL-L &     LDL-C &  2.574 \\
   &     HDL-C &  0.320 \\
   &     IDL-L &   -2.633 \\
  LDL-C, HDL-C, IDL-FC &     LDL-C &  2.244 \\
   &     HDL-C &  0.315 \\
   &     IDL-FC &   -2.271 \\
  LDL-C, HDL-C, IDL-C &     LDL-C &  2.138 \\
   &     HDL-C &  0.329 \\
   &     IDL-C &   -2.162 \\
\bottomrule
\end{tabular}
\end{table}

\section{Review of penalized regression methods for correlated data and clustered regression coefficients} \label{sec: review}

When the columns of $\hat \Pi$ are highly correlated, ordinary weighted least squares estimates of $\bbeta$ tend to be unstable. Penalized regression methods address this issue by introducing regularization to stabilize estimation. 

In this section, we review commonly used penalized regression methods for estimating $\bbeta^*$ in the linear model:
\begin{align}
    \hat \bGamma = \hat \Pi \bbeta^* + \bm \epsilon \label{eq: regression problem}
\end{align}
where $\bm \epsilon = (\epsilon_1,\dots,\epsilon_p)^T$, and $\epsilon_j \sim N(0, \sigma_{Yj}^2)$ for every $j$. For simplicity, we assume that $\hat \Pi$ is measured with negligible error, though we will relax this assumption in our proposed method.

These methods typically solve the following optimization problem:
\begin{align*}
    \mathcal{L}(\bbeta, \lambda) = \frac{1}{2}(\hat \bGamma - \hat \Pi \bbeta)^TW(\hat \bGamma - \hat \Pi \bbeta) + \Omega(\bbeta,\lambda)
\end{align*}
where $W = {\rm diag}(\sigma_{Y1}^{-2},...,\sigma_{Yp}^{-2})$ accounts for heteroskedastic errors, and $\Omega(\bbeta,\lambda)$ is a penalty function of $\bbeta$ and tuning parameter $\lambda$. Well-known examples include LASSO with $\Omega(\bbeta, \lambda) = \lambda||\bbeta||_1$, Ridge with $\Omega(\bbeta, \lambda) = \lambda ||\bbeta||_2^2$ and Elastic Net with $\Omega(\bbeta,\lambda) = \lambda(\alpha||\bbeta||_1 + (1-\alpha)||\bbeta||_2^2)$, where $\alpha$ can be another tuning parameter or user-specified. To incorporate group structure among exposures, the sparse group LASSO adds an additional group-wise $l_2$ penalty: $\Omega(\bbeta,\lambda) = \lambda (\alpha||\bbeta||_1 + (1-\alpha)\sum_{l=1}^L ||\bbeta_l||_2)$ where $\bbeta_l$ is a vector of coefficients for the $l$-th group and $\alpha$ balances the sparsity between and within groups. However, when group structure is unknown, methods that can perform estimation and grouping simultaneously in a data-driven way are more practical.

To address this limitation, \cite{sharma2013consistent} proposed the pairwise absolute clustering and shrinkage (PACS) estimator with:
\begin{align}
    \Omega(\bbeta,\lambda) = \lambda\bigg\{\sum_{k = K} w_k |\beta_k| + \sum_{1\le k < m \le K} w_{km(-)} |\beta_k - \beta_m| + \sum_{1\le k < m\le K} w_{km(+)}|\beta_k + \beta_m|\bigg\} \label{penalty PACS}
\end{align}
where $w_k$, $w_{km(-)}$, and $w_{km(+)}$ are weights that can be user-specified or data-adaptive. For example, the weights are often based on correlations between columns of $\hat \Pi$. This penalty encourages both sparsity and grouping of effects of similar magnitude, without requiring pre-specified groups, a feature that is particularly appealing when highly correlated exposures have indistinguishable effects, as it tends to shrink them to a shared value. 

Other related approaches include OSCAR \citep{bondell2008simultaneous} and the cluster elastic net (CEN) estimator \citep{witten2014cluster}. The former can be considered as a special case of PACS \citep{sharma2013consistent}. CEN uses an $l_1$ penalty together with an $l_2$ within-cluster shrinkage penalty, and combines a data-driven clustering step via $K$-means. Unlike PACS, CEN encourages similarity within clusters but does not enforce exact equality of effects, and it requires the number of clusters to be pre-specified. Refer to \cite{buch2023systematic} for a more comprehensive review of related approaches.

In this work, we adopt the PACS penalty for its ability to identify signal-groups—groups of risk factors with indistinguishable effects—without relying on prior group information.

\section{Estimand based on signal-groups}

As noted in Methods Section 4.8, MVMR-PACS defines a transformation matrix $G$, which converts $\Pi$ to a new design matrix $\Pi_g = \Pi G^T$ with $L$ columns by collapsing elements of $\bgamma_j$, i.e., summing those with nonzero effects of the same magnitude, subtracting those with opposite effects, and omitting those with zero effects. 
, where $L$ is the number of identified groups. Based on this new design matrix for signal-groups, a working model is given by $\bGamma = \Pi_g C_g\bbeta^{*}$, where $\bbeta^{*} = (\beta_1^{*}, ..., \beta_{L}^{*})$ represents the direct causal effects of $L$ signal-groups. The estimand of interest based on the working model is then 
\begin{align*}
    \bbeta_g^{*} = (\Pi_g^TW\Pi_g)^{-1}\Pi_g^TW\bGamma,
\end{align*}
where $W = {\rm diag}(\sigma_{Y1}^{-2},\dots,\sigma_{Yp}^{-2})$.

To demonstrate when this estimand is reasonable to pursue, we consider an example with $K_1 + K_2$ exposures where $K_1, K_2 > 1$. 

Suppose that $\bgamma_j = (\gamma_{j1}, \dots, \gamma_{jK_1},\gamma_{j(K_1+1)},\dots,\gamma_{j(K_1+K_2)})^T \sim N(\bm 0, \Sigma_\gamma)$, where $\Sigma_\gamma$ is a well-defined positive definite matrix. Let $\rho_{kl}$ denote the correlation coefficient for the correlation between $\gamma_{jk}$ and $\gamma_{jl}.$ $\Gamma_j = \bgamma_j \bbeta^* = \sum_{k=1}^{K_1} \gamma_{jk} \beta_k^* + \sum_{k=K_1+1}^{K_2} \gamma_{jk} \beta_k^*$. For simplicity, assume $\sigma_{Yj}^{2}$'s are the same across $j$, all $\rho_{kl} \ge 0$, and diagonal terms of $\Sigma_{\gamma}$ are all 1. For ease of notation, we omit the superscript $^*$ from all betas in the following derivations.

Let $g_{1j} = \sum_{k=1}^{K_1} \gamma_{jk}$ be the SNP-exposure association for the sum of first $K_1$ exposures (signal-group 1), and $g_{2j} = \sum_{k=K_1+1}^{K_2} \gamma_{jk}$ be the SNP-exposure association for the sum of remaining $K_2$ exposures (signal-group 2). Then, the working model based on the grouped exposures is
\begin{align*}
    \Gamma_j & = g_{1j}\tau_1 + g_{2j}\tau_2 + \epsilon_j
\end{align*}
where $\epsilon_j$ measures the difference between the working model and the true model, which is likely correlated with $g_{1j}$ and $g_{2j}$. Because of such correlation, when regressing $\Gamma_j$ on $g_{1j}$ and $g_{2j}$, $\tau_1$ and $\tau_2$ will generally be biased for $\beta$'s.

Our goal here is to compare $\tau_1$, $\tau_2$ with original $\beta$'s, and discuss under which situation grouping may help.

First, we derive the formula for $\tau_1$ and $\tau_2$. Note that
\begin{align*}
    \begin{bmatrix}
        \tau_1 \\
        \tau_2
    \end{bmatrix} = (\Pi_g^T \Pi_g)^{-1} \Pi_g^T\Gamma
\end{align*}
where $\Pi_g = \Pi G^T$ with $G^T = \begin{bmatrix}
    \bm 1_{K_1} & 0\\
    0 & \bm 1_{K2}
\end{bmatrix}$, and
\begin{align*}
    \Pi_g^T \Pi_g = \begin{bmatrix}
        \sum_j g_{1j}^2 & \sum_j g_{1j} g_{2j} \\
        \sum_j g_{1j} g_{2j} & \sum_j g_{2j}^2
    \end{bmatrix}
\end{align*}
and 
\begin{align*}
    \Pi_g^T\Gamma = \begin{bmatrix}
        \sum_j g_{1j} \Gamma_j \\
        \sum_j g_{2j} \Gamma_j
    \end{bmatrix}
\end{align*}

As $p$ is large, we have
\begin{align*}
    & \sum_j g_{1j}^2 = p \var(g_{1j}) = p [K_1 + 2\sum_{k<l}^{K_1}\rho_{kl}] = p [K_1 + C_1]\\
    & \sum_j g_{1j} g_{2j} = p \cov(g_{1j}, g_{2j}) = p \sum_{k=1}^{K_1}\sum_{l=K_1+1}^{K_1+K_2} \rho_{kl} = p D \\
    & \sum_j g_{2j}^2 = p \var(g_{2j}) = p[K_2 + 2\sum_{K_1<k<l<K_1+K_2}^{K_1+K_2} \rho_{kl}] = p [K_2 + C_2] \\
    & \sum_j g_{1j} \Gamma_j = p \cov(g_{1j}, \Gamma_j) = p \cov(\sum_{k=1}^{K_1}\gamma_{jk},\sum_{k=1}^{K_1}\gamma_{jk}\beta_k + \sum_{k=K_1+1}^{K_2}\gamma_{jk} \beta_k) \\
    & = p \bigg(\sum_{k=1}^{K_1}\beta_k + 2\sum_{k<l}^{K_1} \rho_{kl}\beta_l + \sum_{k=1}^{K_1}\sum_{l=K_1+1}^{K_1+K_2} \rho_{kl} \beta_l  \bigg) = p[K1 \bar \beta_1 + A_1 + B_1] \\
    & \sum_j g_{2j} \Gamma_j = p \cov(g_{2j}, \Gamma_j) = p \cov(\sum_{k=K_1+1}^{K_1+K_2}\gamma_{jk},\sum_{k=1}^{K_1}\gamma_{jk}\beta_k + \sum_{k=K_1+1}^{K_2}\gamma_{jk} \beta_k) \\\\
    & = p \bigg(\sum_{k=K_1+1}^{K_1+K_2} \beta_k + 2 \sum_{K_1<k<l\le K_1+K_2}^{K_1+K_2} \rho_{kl}\beta_l + \sum_{k=1}^{K_1} \sum_{l=K_1+1}^{K_1+K_2} \rho_{kl} \beta_k \bigg) = p [K_2 \bar\beta_2+A_2+B_2]
\end{align*}
where $\bar \beta_1 = \frac{1}{K_1} \sum_{k=1}^{K_1} \beta_k$ and $\bar \beta_2 = \frac{1}{K_2} \sum_{k=K_1+1}^{K_1+K_2} \beta_k$.

Thus,
\begin{align*}
    & \Pi_g^T \Pi_g  \approx  \begin{bmatrix}
        p\var(g_{1j}) & p\cov(g_{1j}, g_{2j}) \\
        p\cov(g_{1j}, g_{2j}) & p\var(g_{2j})
    \end{bmatrix} \\
    & (\Pi_g^T \Pi_g)^{-1} = \frac{1}{p(\var(g_{1j})\var(g_{2j}) - \cov(g_{1j}, g_{2j})^2)} \begin{bmatrix}
        \var(g_{2j}) & - \cov(g_{1j}, g_{2j}) \\
        - \cov(g_{1j}, g_{2j}) & \var(g_{1j})
    \end{bmatrix} \\
    & \Pi_g^T\Gamma \approx \begin{bmatrix}
        p\cov(g_{1j}, \Gamma_j) \\
        p\cov(g_{2j}, \Gamma_j)
    \end{bmatrix}
\end{align*}

Therefore, we have
\begin{align*}
    \tau_1 & = \frac{\var(g_{2j})\cov(g_{1j},\Gamma_j) - \cov(g_{1j},g_{2j})\cov(g_{2j},\Gamma_j)}{\var(g_{1j})\var(g_{2j}) - \cov(g_{1j},g_{2j})^2} \\
    & = \frac{(K_2+C_2)(K_1\bar\beta_1+A_1+B_1) - D(K_2\bar\beta_2 + A_2 + B_2)}{(K_1+C_1)(K_2+C_2)- D^2}
\end{align*}
and
\begin{align*}
    \tau_2 & = \frac{\var(g_{1j})\cov(g_{2j},\Gamma_j) - \cov(g_{1j},g_{2j})\cov(g_{1j},\Gamma_j)}{\var(g_{1j})\var(g_{2j}) - \cov(g_{1j},g_{2j})^2} \\
    & = \frac{(K_1+C_1)(K_2\bar\beta_2+A_2+B_2) - D(K_1\bar\beta_1 + A_1 + B_1)}{(K_1+C_1)(K_2+C_2)- D^2}
\end{align*}

Without further assumptions on $\beta$'s and $\rho$'s, we can see that each group's effect is a linear combination of all $\beta$'s in its own group and, through the correlation between the two groups, also dependent on the $\beta$'s in the other group.

However, there are certainly a few special settings where grouping can be beneficial.

Setting 1: $\beta$'s in a group are all equal.

Suppose $\beta_k = \alpha_1^{*}$ for all $k = 1,\dots, K_1$ and $\beta_k = \alpha_2^{*}$ for all $k = K_1+1,\dots, K_1+K_2$. Then,
$\bar\beta_1 = \alpha_1^{*}$, $A_1 = C_1 \alpha_1^{*}$, $B_1 = D\alpha_2^*$, $\bar \beta_2 = \alpha_2^{*}$, $A_2 = C_2 \alpha_2^{*}$, $B_2 = D \alpha_1^{*}$. Then, we have
\begin{align*}
    \tau_1
    & = \frac{(K_2+C_2)(K_1\bar\beta_1+A_1+B_1) - D(K_2\bar\beta_2 + A_2 + B_2)}{(K_1+C_1)(K_2+C_2)- D^2} \\
    & = \frac{(K_2+C_2)(K_1\alpha_1^{*}+C_1\alpha_1^{*}+D\alpha_2^{*}) - D(K_2\alpha_2^{*} + C_2\alpha_2^{*} + D\alpha_1^{*})}{(K_1+C_1)(K_2+C_2)- D^2} \\
    & = \alpha_1^{*}
\end{align*}
Similary, $\tau_2 = \alpha_2^{*}$.
Hence, if the effects of risk factors that we intend to group are indeed equal, then we can still unbiasedly estimate that shared effects.

Setting 2: correlations within the two signal-groups are close to 1 and correlations between the two groups are identical, equal to $r^*$. Then,
$A_1 = (K_1 - 1)\sum_{k=1}^{K_1} \beta_k = K_1^2 \bar \beta_1 - K_1 \bar \beta_1$, $B_1 = r^{*} K_1 K_2 \bar \beta_2$, $C_1 = K_1(K_1-1)$, $A_2 = (K_2-1)\sum_{k=K_1+1}^{K_1+K_2}\beta_k = K_2^2\bar\beta_2 - K_2\bar \beta_2$, $B_2 = r^{*} K_1K_2\bar \beta_1$, $C_2 = K_2(K_2-1)$, $D = r^{*} K_1K_2$. Thus,
\begin{align*}
    \tau_1 & = \frac{K_2^2(K_1^2\bar\beta_1+r^*K_1K_2\bar\beta_2) - r^*K_1K_2(K_2^2\bar\beta_2 + r^*K_1K_2\bar \beta_1)}{K_1^2K_2^2 - r^{*2}K_1^2K_2^2} \\
    & = \bar \beta_1
\end{align*}
Similarly, we can show $\tau_2 = \bar\beta_2$. Hence, when internal correlations are high, the direct causal effect of a signal-group would approximate an average of the causal effects of its group members.

Setting 3: correlations between the two signal-groups are equal to 0, i.e., $\rho_{kl} = 0$ for $k \in \{1,\dots,K_1\}$ and $l \in \{K_1+1,\dots,K_1+K_2\}$. Then, $B_1 = B_2 = 0$ and $D = 0$.
\begin{align*}
    \tau_1 & = \frac{K_1\bar\beta_1 + A_1}{K_1+C_1} \\
    & = \frac{\sum_{k=1}^{K_1} \beta_k + 2 \sum_{k<l}^{K_1}\rho_{kl}\beta_l}{K_1 + 2\sum_{k<l}^{K_1}\rho_{kl}}
\end{align*}
\begin{align*}
    \tau_2 & = \frac{K_2\bar\beta_2 + A_2}{K_2+C_2} \\
    & = \frac{\sum_{k=K_1+1}^{K_1+K_2} \beta_k + 2 \sum_{K_1<k<l\le K_1+K_2}^{K_1+K_2}\rho_{kl}\beta_l}{K_2 + 2\sum_{K_1<k<l\le K_1+K_2}^{K_1+K_2}\rho_{kl}}
\end{align*}
Hence, when the correlations across signal-groups are close to 0, the direct causal effect of a signal-group would approximate a weighted average of the causal effects of its group members.

In the above derivations, we assume that diagonal terms of $\Sigma_{\gamma}$ are all 1. If this is not the case, then above derivations still go through, but the formula becomes more complicated. In particular, in setting 2 where internal correlations are high, the direct causal effect of a signal-group is no longer a simple average of the causal effects of its group members. Instead, it becomes a weighted average of the causal effects of its group members, with weights depending on the variance of $\bgamma_{j}$. We also assume that all $\rho_{kl} \ge 0$. If $\rho_{kl} \to -1$, MVMR-PACS would tend to create a signal-group with the contrast between exposure $k$ and $l$. In this case, we can take the difference between the two exposures and conduct a similar analysis.

\section{MVMR-PACS-x estimator}

MVMR-PACS-$x$ additionally incorporates a correlation threshold of value $x$, e.g., 0.8, into its weighting scheme. Specifically, the weights $w_{km(-)}$ and $w_{km(+)}$ in MVMR-PACS are further multiplier by indicator functions $I(\hat r_{km}>r^*)$ and $I(\hat r_{km}<-r^*)$, respectively. In other words, $w_{km(-)}$ and $w_{km(+)}$ are non-zero, only when the observed correlations exceed a pre-specified threshold $r^*$, e.g., $r^*=0.8$. When all observed correlations fall below this threshold in magnitude, the PACS penalty reduces to the adaptive LASSO penalty \citep{zou2006adaptive}.
Compared with MVMR-PACS, MVMR-PACS-x can achieve improved variable selection performance when the threshold is properly chosen (see the simulation results in Section 2.3). In practice, the threshold can be selected as an additional tuning parameter.

\section{Multifold data thinning for cross validation and post-selection inference}

\subsection{Multifold data thinning for summary-data MVMR}

We adapt the multifold data thinning procedure \citep{neufeld2023data} to summary data MVMR settings to generate multiple independent replicates of summary statistics. Take five-fold data thinning as an example. Let $\epsilon_{1},\dots,\epsilon_5$ be the proportion of information allocated to each generated fold, such that $\sum_{m=1}^5 \epsilon_m = 1$.  For each SNP $j$, $\hat \bgamma_j \sim N(\bgamma_j, \Sigma_{Xj})$. To generate 5 independent SNP-exposure association estimates, we draw
\begin{align*}
    \begin{bmatrix}
        \hat \bgamma_j^{(1)} \\
        \hat \bgamma_j^{(2)} \\
        \hat \bgamma_j^{(3)} \\
        \hat \bgamma_j^{(4)} \\
        \hat \bgamma_j^{(5)}
    \end{bmatrix}  | \hat \bgamma_j 
    \sim N( \begin{bmatrix}
        \epsilon_1 \hat \bgamma_j \\
        \epsilon_2 \hat \bgamma_j \\
        \epsilon_3 \hat \bgamma_j \\
        \epsilon_4 \hat \bgamma_j \\
        \epsilon_5 \hat \bgamma_j
    \end{bmatrix},
    \begin{bmatrix}
        \epsilon_1 (1-\epsilon_1)\Sigma_{Xj} & -\epsilon_1 \epsilon_2 \Sigma_{Xj} & -\epsilon_1 \epsilon_3 \Sigma_{Xj} & -\epsilon_1 \epsilon_4 \Sigma_{Xj} & -\epsilon_1 \epsilon_5 \Sigma_{Xj} \\
        -\epsilon_2 \epsilon_1 \Sigma_{Xj} & \epsilon_2 (1-\epsilon_2) \Sigma_{Xj} & -\epsilon_2 \epsilon_3 \Sigma_{Xj} & -\epsilon_2 \epsilon_4 \Sigma_{Xj} & -\epsilon_2 \epsilon_5 \Sigma_{Xj} \\
        -\epsilon_3 \epsilon_1 \Sigma_{Xj} & -\epsilon_3 \epsilon_2 \Sigma_{Xj} & \epsilon_3(1-\epsilon_3) \Sigma_{Xj} & -\epsilon_3 \epsilon_4 \Sigma_{Xj} & -\epsilon_3 \epsilon_5 \Sigma_{Xj} \\
        -\epsilon_4 \epsilon_1 \Sigma_{Xj} & -\epsilon_4 \epsilon_2 \Sigma_{Xj} & -\epsilon_4\epsilon_3 \Sigma_{Xj} & \epsilon_4(1-\epsilon_4) \Sigma_{Xj} & -\epsilon_4 \epsilon_5 \Sigma_{Xj} \\
        -\epsilon_5 \epsilon_1 \Sigma_{Xj} & -\epsilon_5 \epsilon_2 \Sigma_{Xj} & -\epsilon_5\epsilon_3 \Sigma_{Xj} & -\epsilon_5\epsilon_4 \Sigma_{Xj} & \epsilon_5 (1-\epsilon_5) \Sigma_{Xj} \\
    \end{bmatrix}
\end{align*}
According to Theorem 2 in \citep{neufeld2023data}, we have the following results: (1) for $m = 1,\cdots,5$, $\hat \bgamma_j^{(m)} \sim N(\epsilon_m \bgamma_j, \epsilon_m \Sigma_{Xj})$; (2) $\hat \bgamma_j^{(1)},\cdots, \hat \bgamma_j^{(5)}$ are mutually independent; (3)  $\hat \bgamma_j^{(1)}+\hat \bgamma_j^{(2)}+\cdots+ \hat \bgamma_j^{(5)} = \hat \bgamma_j$. 

The SNP-outcome association estimates can be generated analogously by replacing $\hat \bgamma_j$ and $\Sigma_{Xj}$ with $\hat \Gamma_j$ and $\sigma_{Yj}^2$. The generated SNP-outcome association estimate also satisfies: (1) $\hat \bGamma_j^{(m)} \sim N(\epsilon_m \Gamma_j, \epsilon_m\sigma_{Yj}^2)$, where $\Gamma_j = \bgamma_j \bbeta^*$; (2) $\hat \Gamma_j^{(1)},\cdots, \hat \Gamma_j^{(5)}$ are mutually independent; (3)  $\hat \Gamma_j^{(1)}+\hat \Gamma_j^{(2)}+\cdots+ \hat \Gamma_j^{(5)} = \hat \Gamma_j$. 

This procedure leads to five independent replicates, $D_1,\cdots,D_5$ of summary statistics for $p$ SNPs: $D_m = \bigg \{\hat \bgamma_j^{(m)},\epsilon_m\Sigma_{Xj},\hat \Gamma_j^{(m)},\epsilon_m\sigma_{Yj}^2, j = 1,...,p \bigg \}$. Importantly, within each replicate, the following relationship holds: $E[\hat\Gamma_j^{(m)}] = E[\hat \bgamma_j^{(m)}]\bbeta^*$. Therefore, each independent replicate can be in principle used to identify $\bbeta^*$.

$\epsilon_m$ determines the proportion of instrument strength allocated to the $m$-th fold. Specifically, the population-level instrument strength matrix \citep{wu2024more} based on the $D_m$ is
\begin{align*}
    \epsilon_m^{-1}\sum_{j=1}^p \epsilon_m^2\bgamma_j\bgamma_j^{T} \sigma_{Yj}^{-2} = \epsilon_m \sum_{j=1}^p \bgamma_j\bgamma_j^{T} \sigma_{Yj}^{-2},
\end{align*}
which implies that the instrument strength in the $m$-th fold is $\epsilon_m$ of that in the original data. In our implementation, we always let $\epsilon_m$ equal each other, corresponding to even information split across generated replicates. We refer interested readers to \cite{neufeld2023data} for discussion on the implication of this parameter. In our implementation, data thinning was performed using the R package \textsf{datathin} which is available at https://github.com/anna-neufeld/datathin.

\subsection{Cross validation based on multifold data thinning}\label{alg cross validation}

A cross validation procedure can be conducted based on the generated independent replicates through mutlifold data thinning. For example, a 5-fold cross validation can be conducted as follows:
\begin{enumerate}
  \item Create five independent folds of data $D_1,\cdots,D_5$ by applying five-fold data thinning to the original data $D = \bigg \{\hat \bgamma_j,\Sigma_{Xj},\hat \Gamma_j,\sigma_{Yj}^2, j = 1,...,p \bigg \}$ of summary statistics with $\epsilon_1,\cdots,\epsilon_5$.
  \item For each fold $m=1,\ldots,5$:
  \begin{enumerate}
    \item Define the training set $\mathcal{T}_m = D-D_m$ and validation set $\mathcal{V}_m=D_m$. Here, $D-D_m = \bigg \{\hat \bgamma_j - \hat \bgamma_j^{(m)}, (1-\epsilon_m)\Sigma_{Xj},\hat \Gamma_j - \hat \Gamma_j^{(m)}, (1-\epsilon_m)\sigma_{Yj}^2, j = 1,...,p \bigg \}$, where $\hat \bgamma_j - \hat \bgamma_j^{(m)} \sim N((1-\epsilon_m)\bgamma_j, (1-\epsilon_m)\Sigma_{Xj})$ and $\hat \Gamma_j - \hat \Gamma_j^{(m)} \sim N((1-\epsilon_m)\Gamma_j, (1-\epsilon_m)\sigma_{Yj}^2)$.
    \item Fit the model on $\mathcal{T}_m$ (for each candidate tuning parameter $\lambda$).
    \item Evaluate on $\mathcal{V}_m$ and record the performance metric $L_m(\lambda)$.
  \end{enumerate}
  \item Aggregate performance across folds:
  $\bar L(\lambda)=\tfrac{1}{5}\sum_{m=1}^5 L_m(\lambda)$; choose the best tuning parameter $\hat\lambda$ according to some criteria, e.g., minimizing $L(\lambda)$.
  \item Refit the model on the full dataset $D$ using $\hat\lambda$.
\end{enumerate}

\subsection{Post-selection inference}

To avoid using the same dataset for both model selection and inference, we first perform two-fold data thinning to create two independent datasets, denoted $D^{\rm select}$ and $D^{\rm infer}$, with even information split. Suppose $D^{\rm select}$ serves as the selection dataset for model and variable selection. We then conduct a five-fold cross validation based on $D^{\rm select}$ according to the algorithm described in Supplementary Section \ref{alg cross validation}. This step aims to select tuning parameters and identify signal-groups. In the last step, inference on the identified signal-groups is conducted using SRIVW based on $D^{\rm infer}$.

\section{Computation and Implementation details}

The initial estimates $\tilde {\bbeta}_{\rm init}$ are obtained by minimizing the projected debiased loss function with a ridge penalty
\begin{align}
    \mathcal{L}_{\rm dRidge}(\bbeta, \lambda) = \frac{1}{2} \bbeta^T  (\hat \Pi^T W \hat \Pi - V)_{+} \bbeta - \hat \bGamma^T W \hat \Pi \bbeta + \lambda ||\bbeta||_2^2.
\end{align}
The tuning parameter $\lambda$ is selected through five-fold data thinning. These initial estimates are also used in defining the adaptive weights in the PACS penalty. We prove the consistency and asymptotic normality of this initial estimator in Section XX of the Supplement. The tuning parameter $\lambda$ is selected from a grid of values between $[B\times 10^{-4}, B \times 10^2]$ where $B = (\max(\hat \mu_{n,min}, 0)+ p)^{2/5}$ and $\hat \mu_{n,min}$ is the minimum eigenvalue of sample instrument strength matrix \citep{wu2024more}. Theorem \ref{theo: ridge} provides a sufficient condition for the rate of $B$.

To minimize $\mathcal{L}_{\rm PACS}(\bbeta, \lambda)$, we adopt the local quadratic approximation approach used in \cite{sharma2013consistent}. The main idea is to approximate the non-smooth penalty terms by quadratic functions. For example, by Taylor expansion, for $k = 1,\dots,K$, $|\beta_k| \approx \beta_k^2/(2|\tilde \beta_{k}|) + c$, where $|\tilde \beta_{k}|$ is a given non-zero initial value close to $\beta_{k}$ and $c$ is a term irrelevant to $\beta_k$. $|\beta_k - \beta_m|$ and $|\beta_k + \beta_m|$ can be approximated in a similar manner. This leads to an efficient iterative algorithm with closed-form updates, as follows:
\begin{enumerate}
    \item Initialize $\bbeta^{(0)} = \tilde {\bbeta}_{\rm init}$.
    \item At iteration $(t+1)$, calculate
        \begin{align*}
            \hat \bbeta^{(t+1)} = \bigg\{(\hat \Pi^T W \hat \Pi - V)_{+} + \frac{\lambda}{2} (I_w^{(t)} + D_{(-)}^T I_{w(-)}^{(t)} D_{(-)} + D_{(+)}^T I_{w(+)}^{(t)} D_{(+)} \bigg\}^{-1} \hat \Pi^T W \hat \bGamma
        \end{align*}
    \item 	Let $t = t + 1$ and return to Step 2 until convergence.
\end{enumerate}
In the step 2, 
\begin{align*}
    & I_{w}^{(t)} =  {\rm diag}(\frac{w_k}{|\hat \beta_k^{(t)}|}, j = 1,..,K) \\
    & I_{w(-)}^{(t)} =  {\rm diag}(\frac{w_{km(-)}}{|\hat \beta_{k}^{(t)} - \hat \beta_{m}^{(t)}|}, 1\le k < m \le K) \\
    & I_{w(+)}^{(t)} =  {\rm diag}(\frac{w_{km(+)}}{|\hat \beta_{k}^{(t)} + \hat \beta_{m}^{(t)}|}, 1\le k < m \le K)
\end{align*}
and $D_{(-)}$ and $D_{(+)}$ are $\frac{K(K-1)}{2}\times K$ matrices of $\pm 1$ that encode all pairwise differences and sums, respectively.

We note that when $(\hat \Pi^T W \hat \Pi - V)_{+}$ is positive definite, the objective function $\mathcal{L}_{\rm PACS}(\bbeta, \lambda)$ is strictly convex and the above algorithm is guaranteed to converge to the global minimizer. When it is only positive semi-definite, e.g., in the presence of highly correlated risk factors, then the algorithm will converge to a local minimizer.

We computed $(\hat \Pi^T W \hat \Pi - V)_{+}$ using the \textsf{ADMM\_proj} function from the R package \textsf{BDcocolasso} which is available at https://github.com/celiaescribe/BDcocolasso.

\section{Proofs}

We denote the sample sizes for the outcome and each exposure as $n_Y, n_{Xk}, k=1,\dots, K$, respectively. Let ${\rm diag}(a_1, \dots, a_K)$ denote a $K\times K$ diagonal matrix with diagonal entries $a_1, \dots, a_K$.  We make the following assumption on the observed data.
\begin{assumption}\label{assump: 1}
   (i) The sample sizes $n_Y, n_{Xk}, k=1,\dots, K$ diverge to infinity at the same rate. Let $n=\min (n_Y, n_{X1}, \dots,  n_{XK}) $. The number of exposures $K$ and the number of SNPs $p$ grow to infinity as $n $ increases. \\
   (ii) The $2p$ random vectors (variables) 
    $\{{\hat\bgamma_j, \hat\Gamma_j}, j=1,\dots, p \} $ are mutually  independent. For each $j$, $\hat\bgamma_j \sim N(\bgamma_j, \Sigma_{Xj}), \hat \Gamma_j \sim N(\Gamma_j, \sigma_{Yj}^2) $, $  \Gamma_j =\bgamma_j^T \bm\beta^*$. \\
    (iii)   The variances $\Sigma_{Xj}  = 
	 {\rm diag}(\sigma_{Xj1}, \dots, \sigma_{XjK}) \ \Sigma \  {\rm diag}(\sigma_{Xj1}, \dots, \sigma_{XjK})$ and $\sigma_{Yj}^2$ are known, where $\sigma_{Xjk}$ is the SE of $\hat \gamma_{jk}$ and $\Sigma $ represents a positive definite shared correlation matrix. Moreover, the variance ratios $\sigma_{Xjk}^2/\sigma_{Yj}^2$ are bounded away from zero and infinity for all $j =1,\dots, p$ and $k = 1,...,K$.
\end{assumption}

\begin{assumption}\label{assumption bounded}
    There exists a non-random  $K\times K$ matrix $S_n = \tilde S_n {\rm diag}(\sqrt{\mu_{n1}}, ..., \sqrt{\mu_{nK}})$ that satisfy the following conditions:
    (i) $\tilde S_n$ is bounded element-wise and the smallest eigenvalue of $\tilde S_n \tilde S_n^T$ is bounded away from 0, and
    (ii) $S_n^{-1}\big(\sum_{j=1}^{p} \bgamma_j \bgamma_j^T \sigma_{Yj}^{-2}\big)S_n^{-T}$ is bounded element-wise and its smallest eigenvalue is bounded away from 0.
\end{assumption}

These two assumptions are also made in \cite{wu2024more} with detailed discussions. In particular, Assumption 2 provides a general asymptotic regime to study MVMR estimators, allowing for varying degrees of instrument strengths across different combinations of exposures. The vector $\mu_{n1},\dots,\mu_{nK}$  represents a complete, multi-dimensional summary of overall instrument strength, and the scalar $\mu_{n,\min} = \min (\mu_{n1},\dots, \mu_{nK})$ captures the slowest rate among all linear combinations of exposures. We further assume all the standard error estimates as well as the shared correlation matrix $\Sigma$ can be estimated with negligible errors and hence are considered as known quantities. Additionally, note that $\bgamma_j, \Gamma_j, \Sigma_{Xj},\sigma_{Yj}, j=1,\dots, p$ are sequences of real numbers or matrices that can depend on $n$ in theoretical analysis, but we omit the $n$ index for ease of notation.

\subsection{MVMR-dRidge}

The oracle properties of MVMR-PACS (see Theorem \ref{theorem: PACS}) require a consistent and asymptotic normal initial estimator. Theorem \ref{theo: ridge} establishes the consistency and asymptotic normality of MVMR-dRidge (MVMR debiased ridge estimator), which is used as the initial estimator. MVMR-dRidge is the unique minimizer to the following loss function:
\begin{align}
    \hat \bbeta_{\rm dRidge} = \argmin_{\bbeta} \frac{1}{2} \bbeta^T  (\hat \Pi^T W \hat \Pi - V)_{+} \bbeta - \hat \bGamma^T W \hat \Pi \bbeta + \phi ||\bbeta||_2^2.
\end{align}
$\hat \bbeta_{\rm dRidge}$ has a closed-form solution:
\begin{align*}
    \hat \bbeta_{\rm dRidge} = ((\hat \Pi^T W \hat \Pi - V)_{+} + \phi I_K)^{-1} \hat \Pi W \hat \bGamma.
\end{align*}
\begin{theorem} \label{theo: ridge}
Assume Assumptions 1-2, ${\mu_{n,\min}}/{\sqrt{p}}\rightarrow \infty$, and $\max_j \gamma_{jk}^2\sigma_{Xjk}^{-2}/(\mu_{n,\min} + p) \rightarrow 0$ for all $k$ as $n\to\infty$. If $\phi = o_P(\sqrt{\mu_{n,\min}+p})$, where $a_n = O_P (b_n)$ means $a_n/b_n$ is bounded in probability,
    then $\hat \bbeta_{\rm dRidge}$ is consistent and asymptotically normal, i.e., as $n\to \infty$,
\begin{align}
    \mathbb{V}^{-\frac{1}{2}}\bigg\{ \sum_{j=1}^{p}M_j \bigg\}( \hat \bbeta_{\rm dRidge} - \bbeta_0) \xrightarrow[]{d} N(\bm0, I_{K}), \label{eq: dmvmr normal}
\end{align} 
where $M_j = \bgamma_{j}\bgamma_{j}^T\sigma_{Yj}^{-2}$, $V_j = \Sigma_{Xj}\sigma_{Yj}^{-2}$,   and $\mathbb{V} = \sum_{j=1}^{p}\big\{  (1+ \bm{\beta}_0^TV_j\bm{\beta}_0) (M_j+V_j)+V_j\bm{\beta}_0\bm{\beta}_0^TV_j\big\}$.
\end{theorem}
The proof largely mirrors the proof for Theorem 2 in \citep{wu2024more}, and hence is omitted. The key modification is to show $\phi \mathbb{V}^{-\frac{1}{2}} (\sum_j M_j)^{-1} \mathbb{V}^{\frac{1}{2}} = o_p(1)$ and $\phi \mathbb{V}^{-\frac{1}{2}} = o_p(1)$. In practice, we constructed a logarithmically spaced grid of candidate penalty values for $\phi$, ranging from $10^{2}$ to $10^{-4}$ times an upper bound that scales with $\hat \mu_{n,\rm min}+p$ when $\hat \mu_{n,\rm min} \ge 0$, and scales with $p$ when $\hat \mu_{n,\rm min} < 0$. 

\subsection{MVMR-dLASSO}

MVMR-dLASSO is the unique minimizer to the following loss function:
\begin{align*}
    \hat \bbeta_{\rm dLASSO} & = \argmin_{\bbeta} \frac{1}{2} \bbeta^T(\hat \Pi^T W \hat \Pi - V)_{+} \bbeta - \hat \bGamma^T W \hat \Pi \bbeta + \lambda_n \sum_k \hat w_k |\beta_k|.
\end{align*}
where $\hat w_k = 1/|\tilde \beta_k|^{\tau}$ and $\tilde \bbeta_{\rm init} = (\tilde \beta_1,\dots,\tilde \beta_K)^T$ is a consistent and asymptotically normal estimator of $\bbeta^*$ under the assumptions of Theorem \ref{theo: ridge}.

MVMR-dLASSO is the unique minimizer to the following objective function:
\begin{align*}
    \hat \bbeta_{\rm dLASSO}^{(n)}  = \argmin_{\bbeta} \frac{1}{2} \bbeta^T(\hat \Pi^T W \hat \Pi - V)_{+} \bbeta - \hat \bGamma^T W \hat \Pi \bbeta + \lambda_n \sum_k \hat w_k |\beta_k|
\end{align*}
where $\hat w_k = 1/|\tilde \beta_k|^{\tau}$ and $\tilde \bbeta_{\rm init} = (\tilde \beta_1,\dots,\tilde \beta_K)^T$ is a consistent and asymptotically normal estimator of $\bbeta^*$ under the assumptions of Theorem \ref{theo: ridge}.

\begin{theorem} \label{theorem: dlasso}
    (Oracle properties of MVMR-dLASSO) Under Assumptions 3-\ref{assumption bounded}, ${\mu_{n,\min}}/{\sqrt{p}}\rightarrow \infty$, and $\max_j \gamma_{jk}^2\sigma_{Xjk}^{-2}/(\mu_{n,\min} + p) \rightarrow 0$ for all $k$ as $n\to\infty$, suppose that $\lambda_n/r_n \rightarrow 0$, $\lambda_n r_n^{\tau - 1} \rightarrow \infty$ where $r_n = \mu_{n,\min}/\sqrt{\mu_{n,\min} + p}$ and $\tau$ satisfies $\lambda_nr_n^{\tau+1}/\mu_{n,\max}\rightarrow \infty$, then
    \begin{enumerate}
        \item Consistency in risk factor selection: $\lim_n \probP(\mathcal{A}_n = \mathcal{A}) = 1$
        \item Asymptotic normality: $
    \mathbb{V}_{\mathcal{A}}^{-\frac{1}{2}} (\Pi^T W \Pi)_{\mathcal{A}}(\hat \bbeta_{{\rm dLASSO}, \mathcal{A}}^{(n)} - \bbeta_{\mathcal{A}}^*)  \xrightarrow[]{D} N(\bm 0, I_{|\mathcal{A}|})$
    \end{enumerate}
\end{theorem}
where $\mathcal{A} = \{k: \beta_k^* \neq 0 \}$, $\mathcal{A}_n = \{k: \hat \beta_{{\rm dLASSO},k}^{(n)} \neq 0\}$, and $\mathbb{V} = \sum_{j=1}^{p}\big\{  (1+ \bbeta^{*T}V_j\bbeta^*) (M_j+V_j)+V_j\bbeta\bbeta^{*T}V_j\big\}$, $I_{|\mathcal{A}|}$ is an identity matrix of size $|\mathcal{A}|\times|\mathcal{A}|$, and the subscript $\mathcal{A}$ applied to a vector (or matrix) denotes the subvector (or submatrix) formed by selecting the entries (or rows and columns) corresponding to the indices in the set $\mathcal{A}$.

\begin{proof}
    The proof follows the proof of adaptive LASSO in \cite{zou2006adaptive} with adaption of a random design matrix and non-i.i.d. setup. To ease the notation, we remove the subscript $\rm dLASSO$ in the objective function and MVMR-dLASSO estimator in this proof. Under Assumption \ref{assumption bounded} and by Lemma 7 in \cite{wu2024more} that $S_n^{-1}(\hat \Pi^T W \hat \Pi - \Pi^TW \Pi - V)S_n^{-T} \xrightarrow[]{P} \bm 0$, we know the minimum eigenvalue of $\hat \Pi^T W \hat \Pi - V$ is $\Theta_p(\mu_{n,\min})$. Because we have assumed that $\mu_{n,\min}/\sqrt{p} \rightarrow \infty$, $(\hat \Pi^T W \hat \Pi - V)_{+} = \hat \Pi^T W \hat \Pi - V$ for large $n$. Therefore, for large $n$, the objective function for MVMR-dLASSO can be simplified as
    \begin{align*}
    \Psi^{(n)} (\bbeta) = \frac{1}{2} \bbeta^T(\hat \Pi^T W \hat \Pi - V) \bbeta - \hat \bGamma^T W \hat \Pi \bbeta + \lambda_n \sum_k \hat w_k |\beta_k|
    \end{align*}

    Define $\bm u = r_n(\bbeta - \bbeta^*)$. Then, $\bbeta = \bbeta^{*} + \frac{\bm u}{r_n}$, where under the assumption $\mu_{n,\min}/\sqrt{p} \rightarrow \infty$, $r_n \rightarrow \infty$. Then, the objective function can be writen as a function of $\bm u$ as follows
    \begin{align*}
        \Psi^{(n)}(\bm u) & = \frac{1}{2} (\bbeta^* + \frac{\bm u}{r_n})^T  (\hat \Pi^T W \hat \Pi - V) (\bbeta^* + \frac{\bm u}{r_n}) - \hat \bGamma^T W \hat \Pi (\bbeta^* + \frac{\bm u}{r_n})  +  \lambda_n \sum_{k=1} \hat w_k |\beta_k^* + \frac{u_k}{r_n}|
    \end{align*}

Let $\hat {\bm u}^{(n)} = \argmin \Psi^{(n)}(\bm u)$. The goal is then to figure out the asymptotic properties of $\hat {\bm u}^{(n)}$. Once we know the properties of $\hat {\bm u}^{(n)}$, we know the properties of $\hat \bbeta^{(n)}$ because  $\hat \bbeta^{(n)} = \bbeta^* + \frac{ \hat {\bm u}^{(n)}}{r_n}$. 

Define $V^{(n)}(\bm u) = \Psi_n(\bm u) - \Psi_n(\bm 0)$. Specifically, 
\begin{align*}
    V^{(n)}(\bm u) & = \frac{1}{2r_n^2} \bm u^T (\hat \Pi^T W \hat \Pi- V) \bm u - \frac{1}{r_n}\bm u^T (\hat\Pi^T W \hat \bGamma - \hat \Pi^T W \hat \Pi \bbeta^* + V \bbeta^*) \\
    & + \lambda_n \sum_{k=1} \hat w_k \{|\beta_k^* + \frac{u_k}{r_n}| - |\beta_k^*| \}
\end{align*}

For the first term, by assumption \ref{assumption bounded} and Lemma 7 in \cite{wu2024more} that $S_n^{-1} (\hat \Pi^T W \hat \Pi - \Pi^T W \Pi - V) S_n^{-T} \xrightarrow[]{P} \bm 0$, the minimum and maximum eigenvalue of $\hat \Pi^T W \hat \Pi - V$ are of order $\Theta(\mu_{n,\min})$ and $\Theta(\mu_{n,max})$, respectively. Then, for every $\bm u$, we have 
\begin{align}
    \frac{1}{r_n^2}\bm u^T (\hat \Pi^T W \hat \Pi - V) \bm u = O_p(\frac{\mu_{n,\max}}{r_n^2}), \label{rate of squared term in obj}
\end{align}
and $\frac{1}{r_n^2}(\hat \Pi^T W \hat \Pi - V)$ remains positive definite with high probability as $n \rightarrow \infty$ because the minimum eigenvalue of $\frac{1}{r_n^2}(\hat \Pi^T W \hat \Pi - V) = \Theta_p(\frac{\mu_{n,\min}+p}{\mu_{n,\min}}) \ge \Theta_p(1)$. This ensures that $V^{(n)}(\bm u)$ is a strictly convex function of $\bm u$ for large $n$, and $\hat {\bm u}^{(n)}$ is unique and bounded. 

For the third term, we consider its behavior as $n \rightarrow \infty$. For $\beta_k^* \neq 0$, we have $\hat w_k \xrightarrow[]{P} |\beta_k^*|^{-\tau}$ by the assumption that $\tilde \bbeta_{\rm init}$ is consistent, and $r_n(|\beta_k^* + \frac{u_k}{r_n}| - |\beta_k^*|) \rightarrow u_k {\rm sign}(\beta_k^*)$. Hence, $\lambda_n \hat w_k (|\beta_k^* + \frac{u_k}{r_n}| - |\beta_k^*|) = O_p(\frac{\lambda_n}{r_n})u_k{\rm sign}(\beta_k^*), \forall k \in \mathcal{A}$. Under the assumption that $\lambda_n/r_n \rightarrow 0$, we have $\lambda_n \hat w_k\{|\beta_k^* + \frac{u_k}{r_n}| - |\beta_k^*| \} = \frac{\lambda_n}{r_n} \hat w_k \bigg[r_n\{|\beta_k^* + \frac{u_k}{r_n}| - |\beta_k^*| \}\bigg] = o_p(1)$. For $\beta_k^* = 0$,  $ r_n\{|\beta_k^* + \frac{u_k}{r_n}| - |\beta_k^*|\} = |u_k|$. Because $\tilde \bbeta_{\rm init}$ is also asymptotically normal, $[\mathcal{V}_n]_{kk}^{-1} \tilde \beta_k \xrightarrow[]{D} N(0,1) = O_p(1)$, where $\mathcal{V}_n$ is a consistent estimator of the asymptotic variance of $\tilde \bbeta_{\rm init}$. Thus, $\lambda_n \hat w_k |\frac{u_k}{r_n}| = \frac{\lambda_n}{r_n}[\mathcal{V}_n]_{kk}^{-\frac{\tau}{2}}  |u_k|  \{ [\mathcal{V}_n]_{kk}^{-\frac{1}{2}} \tilde \beta_k \}^{-\tau}$. The slowest rate at which the diagonal term of $\mathcal{V}_n$ goes to 0 is $r_n^{-2}$ by Lemma 8 in \cite{wu2024more}. By the assumption that $\lambda_n r_n^{\tau - 1} \rightarrow \infty$, we have $ \lambda_n \hat w_k\{|\beta_k^* + \frac{u_k}{r_n}| - |\beta_k^*|\}\rightarrow \infty$ with probability tending one. This holds true unless $\hat u_k^{(n)}$ goes to 0 sufficiently fast, faster than $\lambda_n r_n^{\tau - 1}$.

We now show that $\hat {\bm u}_{\mathcal{A}^c}^{(n)} \xrightarrow[]{P} \bm 0$. Suppose for contradiction that there exists $\epsilon > 0$ and an index $k \in \mathcal{A}^c$, such that, the event $|\hat u_k^{(n)}| \ge \epsilon$ happens infinitely often. This implies that $\lambda_n \hat w_k \{|\beta_k^* + \frac{\hat u_k^{(n)}}{r_n}| - |\beta_k^*|\} \ge \frac{\lambda_n}{r_n} \hat w_k \epsilon$ with probability tending one. As discussed above, when $\beta_k^* = 0$, $\frac{\lambda_n}{r_n} \hat w_k = O_p(\lambda_n r_n^{\tau - 1}) \to \infty$. Furthermore, by the result \eqref{rate of squared term in obj} and boundedness of $\hat {\bm u}^{(n)}$, the rate $\frac{1}{r_n^2}{\hat {\bm u}}^{(n)T}(\hat \Pi^T W\hat \Pi - V){\hat {\bm u}}^{(n)}$ diverges to infinity is at most $\mu_{n,\max}/r_n^2$. By the assumption that $\lambda_n r_n^{\tau + 1}/\mu_{n,\max} \rightarrow \infty$, we have $\lambda_n \hat w_k \{|\beta_k^* + \frac{u_k^{(n)}}{r_n}| - |\beta_k^*|\}$ goes to infinity at a rate faster than $\frac{1}{r_n^2}{\hat {\bm u}}^{(n)T}(\hat \Pi^T W\hat \Pi - V){\hat {\bm u}}^{(n)}$ for any $\hat {\bm u}^{(n)}$ that has $|\hat u_k^{(n)}| \ge \epsilon$ happens infinitely often. Consequently, the objective function $\Psi_n(\hat {\bm u}^{(n)})$ goes to infinity with probability tending one. However, this contradicts the definition of $\hat {\bm u}^{(n)}$ as the minimizer of $\Psi^{(n)}$ because we could always pick $\hat {\bm u}^{(n)}$ with $\hat u_k^{(n)} = 0, \forall k \in \mathcal{A}^c$ and other coordinates suitably to make the objective function smaller. Hence, we must have $\hat u_k^{(n)} \xrightarrow[]{P} 0, \forall k \in \mathcal{A}^c$.

As a result, for large enough $n$, we only need to minimize 
\begin{align*}
    \frac{1}{2 r_n^2} \bm u_{\mathcal{A}}^T (\hat \Pi^T W\hat \Pi- V)_{\mathcal{A}} \bm u_{\mathcal{A}} - \frac{1}{r_n} \bm u_{\mathcal{A}}^T(\hat \Pi^T W\hat \bGamma - \hat \Pi^T W \hat \Pi \bbeta^* + V \bbeta^*)_{\mathcal{A}} + O_p(\frac{\lambda_n}{r_n})\bm u_{\mathcal{A}}^T\cdot{\rm sign}(\bbeta_{\mathcal{A}}^*).
\end{align*}

This objective is convex and the minimizer is
\begin{align*}
    \hat{\bm u}_{\mathcal{A}}^{(n)} = r_n (\hat\Pi^T W \hat \Pi-V)_{\mathcal{A}}^{-1}(\hat \Pi^T W \hat \bGamma - \hat \Pi^T W \hat \Pi \bbeta^* + V\bbeta^*)_{\mathcal{A}} + O_p( \lambda_n r_n)(\hat\Pi^T W \hat \Pi-V)_{\mathcal{A}}^{-1}{\rm sign}(\bbeta_{\mathcal{A}}^*)
\end{align*}

This implies that, 
\begin{align*}
    (\hat \bbeta_{\mathcal{A}}^{(n)} - \bbeta_{\mathcal{A}}^*) & = (\hat\Pi^T W\hat\Pi - V)_{\mathcal{A}}^{-1} (\hat \Pi^T W \hat \bGamma - \hat \Pi^T W \hat \Pi\bbeta^* + V \bbeta^*)_{\mathcal{A}}  + O_p( \lambda_n) (\hat\Pi^T W \hat \Pi-V)_{\mathcal{A}}^{-1}{\rm sign}(\bbeta_{\mathcal{A}}^*),
\end{align*}
and then,
\begin{align*}
    \mathbb{V}_{\mathcal{A}}^{-\frac{1}{2}} (\Pi^T W \Pi)_{\mathcal{A}}(\hat \bbeta_{\mathcal{A}}^{(n)} - \bbeta_{\mathcal{A}}^*) & = \mathbb{V}_{\mathcal{A}}^{-\frac{1}{2}} (\Pi^T W \Pi)_{\mathcal{A}} (\hat\Pi^T W\hat\Pi - V)_{\mathcal{A}}^{-1} (\hat \Pi^T W \hat \bGamma - \hat \Pi^T W \hat \Pi\bbeta^* + V \bbeta^*)_{\mathcal{A}}  \\
    & + O_p( \lambda_n r_n)(\hat\Pi^T W \hat \Pi-V)_{\mathcal{A}}^{-1}{\rm sign}(\bbeta_{\mathcal{A}}^*) \\
    & = \underbrace{\mathbb{V}_{\mathcal{A}}^{-\frac{1}{2}} (\Pi^T W \Pi)_{\mathcal{A}} (\hat\Pi^T W\hat\Pi - V)_{\mathcal{A}}^{-1} \mathbb{V}_{\mathcal{A}}^{\frac{1}{2}}}_{A_1} \underbrace{ \mathbb{V}_{\mathcal{A}}^{-\frac{1}{2}} (\hat \Pi^T W \hat \bGamma - \hat \Pi^T W \hat \Pi\bbeta^* + V \bbeta^*)_{\mathcal{A}} }_{A_2}\\
    & + O_p( \lambda_n) \underbrace{\mathbb{V}_{\mathcal{A}}^{-\frac{1}{2}} (\Pi^T W \Pi)_{\mathcal{A}} (\hat\Pi^T W \hat \Pi-V)_{\mathcal{A}}^{-1}\mathbb{V}_{\mathcal{A}}^{\frac{1}{2}}}_{A_1} \mathbb{V}_{\mathcal{A}}^{-\frac{1}{2}}{\rm sign}(\bbeta_{\mathcal{A}}^*)
\end{align*}
By the result S10 and Lemma 6 in \cite{wu2024more}, we know $A_1 \xrightarrow[]{P} I_{|\mathcal{A}|}$ and $A_2 \xrightarrow[]{D} N(\bm 0, I_{|\mathcal{A}|})$. Additionally, by Lemma 2 in \cite{wu2024more}, $\mathbb{V}_{\mathcal{A}}^{-\frac{1}{2}}$ goes to 0 at a rate at least $\sqrt{\mu_{n,\min} + p}$. Thus, the second term is of order $O_p(\lambda_n/\sqrt{\mu_{n,\min}+p}) = o_p(1)$ as $\lambda_n/r_n \rightarrow 0$. Lastly, by the Slutsky's theorem, we complete the proof of the asymptotic normality part. That is,
\begin{align}
    \mathbb{V}_{\mathcal{A}}^{-\frac{1}{2}} (\Pi^T W \Pi)_{\mathcal{A}}(\hat \bbeta_{\mathcal{A}}^{(n)} - \bbeta_{\mathcal{A}}^*) \xrightarrow[]{D} N(0,I_{|\mathcal{A}|}) = O_p(1). \label{dlasso asymp normal}
\end{align}

To show the selection consistency part, note that $\forall k \in \mathcal{A}$, by the asymptotic normality result, we have $\hat \beta_k^{(n)}\xrightarrow[]{P}\beta_k^*$. Thus, $\probP(k \in \mathcal{A}_n) \rightarrow 1$. It remains to show that $\forall k^{'} \in \mathcal{A}^c$, $\probP(k^{'} \in \mathcal{A}_n) \rightarrow 0$. Consider the event $k^{'} \in \mathcal{A}_n$ (i.e. $\hat \beta_{k^{'}}^{(n)} \neq 0$), by the KKT conditions, we know $\hat \bbeta^{(n)}$ satisfies
\begin{align}
    \bigg [(\hat \Pi^T W \hat \Pi - V) \hat \bbeta^{(n)} - \hat \Pi^T W \hat \bGamma\bigg]_{k^{'}}= \lambda_n \hat w_{k^{'}} {\rm sign} (\hat \beta_{k^{'}}^{(n)}) \label{KKT eq}
\end{align}
where $[\bm v]_k$ means the $k$-th entry of the vector $\bm v$. The right hand side $\lambda_n \hat w_{k^{'}} \rightarrow \infty$ at a rate $\lambda_n r_n^{\tau}$ as we have shown above. For the left hand side, note that
\begin{align*}
    (\hat \Pi^T W \hat \Pi - V) \hat \bbeta^{(n)} - \hat \Pi^T W \hat \bGamma & = (\hat \Pi^T W\hat \Pi - V) (\hat \bbeta^{(n)} - \bbeta^*) - \mathbb{V}^{\frac{1}{2}} \bigg \{\mathbb{V}^{-\frac{1}{2}}(\hat\Pi^T W \hat \bGamma - \hat \Pi^T W\hat \Pi \bbeta^* + V \bbeta^*)\bigg \}
\end{align*}
From Lemma 6 in \cite{wu2024more}, we know $\mathbb{V}^{-\frac{1}{2}}(\hat\Pi^T W \hat \bGamma - \hat \Pi^T W\hat \Pi \bbeta^* + V \bbeta^*) = O_p(1)$ and hence the fast rate at which $\mathbb{V}^{\frac{1}{2}} \bigg \{\mathbb{V}^{-\frac{1}{2}}(\hat\Pi^T W \hat \bGamma - \hat \Pi^T W\hat \Pi \bbeta^* + V \bbeta^*)\bigg \}$ diverges is $\sqrt{\mu_{n,\max} + p}$. For the first term, we study its $k^{'}$ entry. Let $v_k$ denote the $k$-th row of $\hat \Pi^T W \hat \Pi -V$. We know $v_k = O_p(\mu_{n,\max})$ by Assumption \ref{assumption bounded} and Lemma 7 in \cite{wu2024more}. Then, the $k^{'}$ entry of the first term equals
\begin{align*}
     & v_{k^{'}, \mathcal{A}} (\hat \bbeta^{(n)} - \bbeta^{*})_{\mathcal{A}}  + v_{k^{'}, \mathcal{A}^c} (\hat \bbeta^{(n)} - \bbeta^{*})_{\mathcal{A}^c} = v_{k^{'}, \mathcal{A}} (\Pi^T W \Pi)_{\mathcal{A}}^{-1}\mathbb{V}_{\mathcal{A}}^{\frac{1}{2}} \mathbb{V}_{\mathcal{A}}^{-\frac{1}{2}} (\Pi^T W \Pi)_{\mathcal{A}}  (\hat \bbeta^{(n)} - \bbeta^{*})_{\mathcal{A}} + v_{k^{'}, \mathcal{A}^c} \hat \bbeta^{(n)}_{\mathcal{A}^c}.
\end{align*}
We have shown in \eqref{dlasso asymp normal} that $\mathbb{V}_{\mathcal{A}}^{-\frac{1}{2}} (\Pi^T W \Pi)_{\mathcal{A}}  (\hat \bbeta^{(n)} - \bbeta^{*})_{\mathcal{A}} = O_p(1)$. Additionally, $(\Pi^T W \Pi)_{\mathcal{A}}^{-1}\mathbb{V}_{\mathcal{A}}^{\frac{1}{2}}$ goes to 0 at a rate of at least $r_n$ by Lemma 8 in \cite{wu2024more}. Thus, the first term is $O_p(\mu_{n,\max}/r_n)$. The second term is also $O_p(\mu_{n,\max}/r_n)$ as we have shown that $\hat {\bm u}_{\mathcal{A}^c} \xrightarrow[]{P} \bm 0$ and $\hat\bbeta_{\mathcal{A}^c}^{(n)} = \hat {\bm u}_{\mathcal{A}^c}/r_n$.  In summary, $\bigg [(\hat \Pi^T W \hat \Pi - V) \hat \bbeta^{(n)} - \hat \Pi^T W \hat \bGamma\bigg]_{k^{'}}$ is $O_p(\max\{\mu_{n,\max}/r_n, \sqrt{\mu_{n,\max} + p}\}) = O_p(\mu_{n,\max}/r_n)$, because it is easy to see that $\mu_{n,\max}/\sqrt{\mu_{n,\max}+p} \ge \mu_{n,\min}/\sqrt{\mu_{n,\min} + p} = r_n$. Finally, because  $\lambda_n r_n^{\tau+1}/\mu_{n,\max} \rightarrow \infty$, we conclude that the equation \eqref{KKT eq} cannot hold because the right hand side grows faster than the left hand side as $n\to \infty$. This is a contradiction and therefore we must have $\forall k^{'} \in \mathcal{A}^c$, $\probP(k^{'} \in \mathcal{A}_n) \rightarrow 0$.

\end{proof}

\subsection{MVMR-PACS}

MVMR-PACS is the unique minimizer to the following objective function:
\begin{align*}
    \hat \bbeta_{\rm PACS}^{(n)} = \argmin_{\bbeta} \frac{1}{2} \bbeta^T(\hat \Pi^T W \hat \Pi - V)_{+} \bbeta - \hat \bGamma^T W \hat \Pi \bbeta + \lambda_n P(\bbeta)
\end{align*}
where $P(\bbeta) = \sum_{k=1}^K \hat w_k |\beta_k| + \sum_{k < m} \hat w_{km(-)} |\beta_k - \beta_m| + \sum_{k<m} \hat w_{km (+)} |\beta_k + \beta_m|$, 
$\hat w_k = 1/|\tilde \beta_k|^{\tau}$, $\hat w_{km(-)} = \hat c_{km(-)}/|\tilde \beta_k - \tilde \beta_m|^{\tau}$, $\hat w_{km(+)} = \hat c_{km(+)}/|\tilde \beta_k + \tilde \beta_m|^{\tau}$, and $\tilde \bbeta_{\rm init} = (\tilde \beta_1,\dots,\tilde \beta_K)^T$ is a consistent and asymptotically normal estimator of $\bbeta^*$ under the assumption of Theorem \ref{theo: ridge}, $c_{km(-)}^{(n)}$ and $c_{km(+)}^{(n)}$ are user-specified and potentially data-dependent weights. In our setting, $c_{km(-)}^{(n)}$ and $c_{km(+)}^{(n)}$ are given in the main text section 4.4. We require $\hat c_{km(-)}\xrightarrow[]{P} c_{km(-)}$, and $\hat c_{km(+)}\xrightarrow[]{P} c_{km(+)}$ that satisfy $0 < c_{km(-)}, c_{km(+)}<\infty$ for all $k,m \in \{1,\dots,K\}$.

Let $\bar \bgamma_j = G \bgamma_j$ be the true association of $j$-th SNP with $L$ causal groups, where $G$ follows the definition given in the main text section 4.8. In words, $G$ collapses elements of $\bgamma_j$ by taking the sum (difference) if the corresponding true causal effects are nonzero and of the same (opposite) magnitude and it drops elements of $\bgamma_j$ if the corresponding true causal effects are zero.

\begin{lemma}
   For the matrices $G$ and $C_g$ satisfying the definition given in the main text section 4.8, $G^T C_g \bbeta^* = \bbeta^*$.
\end{lemma}
\begin{proof}
    Recall that  for each risk factor $k$, if it belongs to a causal group $G_l$ (i.e. $k\in G_l$ for some $l$), then $\beta_k^* = s_k \alpha_l$, where $s_k \in \{-1, +1\}$ and $\alpha_l > 0$ represents the magnitude of the causal effect of group $G_l$; otherwise, $\beta_k^* = 0$, i.e., it is a non-causal risk factor.

    Next, note that for every $k,j \in \{1,...,K\}$, 
    \begin{align*}
        (G^T C_g)_{(k,j)} = \sum_{l=1}^L G_{kl}C_{g,jl} = \begin{cases}
        \frac{s_k s_{j}}{|G_l|}, & \text{if $\exists l(k) \in \{1,...,L\}$ \ s.t.\ $k,j \in G_{l(k)}$} \\
        0, & \text{otherwise}
    \end{cases}
    \end{align*}
where $l(k)$ denotes the unique group index such that $k \in G_{l(k)}$. In other words, $(G^T C_g)_{(k,j)}$ is only non-zero when $k,l$ belong to the same causal group, and zero otherwise.

Now we consider $G^T C_g \bbeta^*$. Its $k$-th component is
\begin{align*}
    (G^T C_g \bbeta^*)_k & = \sum_{j = 1}^K (G^T C_g)_{kj} \beta_j^*
\end{align*}
There are two scenarios. First, if risk factor $k$ is non-causal, then $(G^TC_g)_{kj} = 0$ for all $j = 1,...,K$. Then, $(G^T C_g \bbeta^*)_k = \beta_k^* = 0$. Second, if risk factor $k$ is causal, then there exists $l(k)$ such that $k \in G_{l(k)}$ and
\begin{align*}
    (G^T C_g \bbeta^*)_k & = \sum_{j = 1}^K (G^T C_g)_{kj} \beta_j^* \\
    & = \sum_{j\in G_{l(k)}} (G^T C_g)_{kj} \beta_j^* \text{ as $(G^T C_g)_{kj}$ is non-zero only if $k,j$ belong to the same causal group}\\
    & = \sum_{j\in G_{l(k)}} \frac{s_ks_j}{|G_{l(k)}|} \beta_j^* \\
    & = \sum_{j\in G_{l(k)}} \frac{s_ks_j}{|G_{l(k)}|} s_j \alpha_{l(k)} \text{ as $\beta_j^* = s_j \alpha_{l(k)}$}\\
    & = \sum_{j\in G_{l(k)}} \frac{s_k}{|G_{l(k)}|} \alpha_{l(k)} \\
    & = s_k \alpha_{l(k)} \sum_{j\in G_{l(k)}} \frac{1}{|G_{l(k)}|}  \\
    & = s_k \alpha_{l(k)} = \beta_k^*
\end{align*}
This completes the proof.
\end{proof}

\begin{theorem} \label{theorem: PACS supp}
    (Oracle properties of MVMR-PACS)  Under the same conditions as in Theorem \ref{theorem: dlasso}, the MVMR-PACS estimator satisfies
    \begin{enumerate}
        \item Consistency in risk factor selection: $\lim_n \probP(\mathcal{B}_n = \mathcal{B}) = 1$
        \item Asymptotic normality: $
    \mathbb{V}_{g}^{-\frac{1}{2}} (\Pi_g^T W \Pi_g)(C_g\hat \bbeta_{{\rm PACS}}^{(n)} - C_g\bbeta^*)  \xrightarrow[]{D} N(\bm 0, I_{L})$
    \end{enumerate}
\end{theorem}
where $\mathcal{B} = \{k: \theta_k \neq 0 \}$, $\mathcal{B}_n = \{k: \hat \beta_{{\rm PACS},k}^{(n)} \neq 0\}$, and $\Pi_g =  \Pi G^T$, $V_{j, g} = G \Sigma_{Xj}\sigma_{Yj}^{-2}G^T$, $\mathbb{V}_g = \sum_{j=1}^{p}\big\{  (1+ \bar \bbeta^{*T}V_{j,g}\bar\bbeta^*) (G\bgamma_j\bgamma_j^T\sigma_{Yj}^{-2}G^T+ V_{j,g})+V_{j,g}\bar \bbeta\bar \bbeta^{*T}V_{j,g}\big\}$, $I_{L}$ is an identity matrix of size $L\times L$.

\begin{proof}
    We prove the asymptotic normality result first. Following the same argument as in the proof of Theorem \ref{theorem: dlasso}, the objective function for MVMR-PACS (with light abuse of notation) can be simplified as
\begin{align*}
    \Psi^{(n)} (\bbeta) = \frac{1}{2} \bbeta^T(\hat \Pi^T W \hat \Pi - V) \bbeta - \hat \bGamma^T W \hat \Pi \bbeta + \lambda_n P(\bbeta)
\end{align*}
Define $\bm u = r_n(\bbeta - \bbeta^*)$. Then, $\bbeta = \bbeta^{*} + \frac{\bm u}{r_n}$, where under the assumption $\mu_{n,\min}/\sqrt{p} \rightarrow \infty$, $r_n \rightarrow \infty$. Then, the objective function can be writen as a function of $\bm u$ as follows
    \begin{align*}
        \Psi^{(n)}(\bm u) & = \frac{1}{2} (\bbeta^* + \frac{\bm u}{r_n})^T  (\hat \Pi^T W \hat \Pi - V) (\bbeta^* + \frac{\bm u}{r_n}) - \hat \bGamma^T W \hat \Pi (\bbeta^* + \frac{\bm u}{r_n})  +  \lambda_n P(\bbeta^* + \frac{\bm u}{r_n})
    \end{align*}

Let $\hat {\bm u}^{(n)} = \argmin \Psi^{(n)}(\bm u)$; then $\hat \bbeta^{(n)} = \bbeta^* + \frac{ \hat {\bm u}^{(n)}}{r_n}$. Let $V^{(n)}(\bm u) = \Psi_n(\bm u) - \Psi_n(\bm 0)$. We know $\hat {\bm u}^{(n)}$ is also the minimizer of $V^{(n)}(\bm u)$, where 
\begin{align*}
    V^{(n)}(\bm u) & = \frac{1}{2r_n^2} \bm u^T (\hat \Pi^T W \hat \Pi- V) \bm u - \frac{1}{r_n}\bm u^T (\hat\Pi^T W \hat \bGamma - \hat \Pi^T W \hat \Pi \bbeta^* + V \bbeta^*) \\
    & + \lambda_n \tilde P(u)
\end{align*}
where 
\begin{align*}
    \tilde P(\bm u) & = \sum_k \hat w_k \{|\beta_k^* + \frac{u_k}{r_n}| - |\beta_k^*|\} + \sum_{1\le k < l \le K} \hat  w_{kl(-)} \{|\beta_k^* - \beta_l^* + \frac{u_k - u_l}{r_n}| - |\beta_k^* - \beta_l^*| \} \\
    & + \sum_{1\le k < l \le K}\hat w_{kl(+)}\{|\beta_k^* + \beta_l^* + \frac{u_k + u_l}{r_n}| - |\beta_k^* + \beta_l^*| \}
\end{align*}
We now consider the limiting behavior of $\lambda_n \tilde P(\bm u)$. We aim to show that $\lambda_n \tilde P(\bm u)$ goes to 0 under correct sparsity and group structure and goes to infinity under the incorrect structure. 

For all $k, l \in \{1,...,K\}$, if $\beta_k^* \neq 0, \beta_l^* \neq 0$ and $\beta_k^* \neq \beta_l^*$, then $\hat w_k \xrightarrow[]{P} |\beta_k^*|^{-\tau}$, $\hat w_{kl(-)} \xrightarrow[]{P} c_{kl(-)} |\beta_k^* - \beta_l^*|^{-\tau}$ and $\hat w_{kl(+)} \xrightarrow[]{P} c_{kl(+)} |\beta_k^* + \beta_l^*|^{-\tau}$. Also, $r_n \{ |\beta_k^* + \frac{u_k}{r_n}| - |\beta_k^*|\} \to u_k{\rm sign}(\beta_k^*)$, $r_n \{ | \beta_k^* - \beta_l^* + \frac{u_k - u_l}{r_n}| - |\beta_k^* - \beta_l^*| \} \to (u_k - u_l) {\rm sign} (\beta_k^* - \beta_l^*)$, and $r_n \{ | \beta_k^* + \beta_l^* + \frac{u_k + u_l}{r_n}| - |\beta_k^* + \beta_l^*| \} \to (u_k + u_l) {\rm sign} (\beta_k^* + \beta_l^*)$. By Slutsky's theorem, we have $\lambda_n \tilde P(\bm u) \xrightarrow[]{P} 0$ because $\lambda_n /r_n \to 0$. 

If $\beta_k^* = 0$, then $r_n\{|\beta_k^* +\frac{u_k}{r_n}| - |\beta_k^*|\} = |u_k|$ and $ \lambda_n \hat w_k\{|\beta_k^* + \frac{u_k}{r_n}| - |\beta_k^*|\} = \frac{\lambda_n}{r_n} |u_k| [\mathcal{V}_n]_{kk}^{-\frac{\tau}{2}} \{[\mathcal{V}_n]_{kk}^{-\frac{1}{2}} \tilde \beta_k \}^{-\tau} =\frac{\lambda_n}{r_n} |u_k| [\mathcal{V}_n]_{kk}^{-\frac{\tau}{2}} O_p(1)$, where the second equality is because $\tilde \bbeta_{\rm init}$ is asymptotically normal, and hence $[\mathcal{V}_n]_{kk}^{-1} \tilde \beta_k \xrightarrow[]{D} N(0,1) = O_p(1)$, where $\mathcal{V}_n$ is a consistent estimator of the asymptotic variance of $\tilde \bbeta_{\rm init}$. Because the slowest rate at which the diagonal term of $\mathcal{V}_n$ goes to $0$ is $r_n^{-2}$ by Lemma 8 in \cite{wu2024more}, we know $\frac{\lambda_n}{r_n} |u_k| [\mathcal{V}_n]_{kk}^{-\frac{\tau}{2}} O_p(1) = O_p(\lambda_n r_n^{\tau - 1}) |u_k|$. By the assumption that $\lambda_n r_n^{\tau - 1} \rightarrow \infty$, we have $\lambda_n \hat w_k\{|\beta_k^* + \frac{u_k}{r_n}| - |\beta_k^*|\} \to \infty$ with probability tending one unless $u_k$ goes to 0 sufficiently fast, faster than $\lambda_n r_n^{\tau - 1}$.

Similarly, if $\beta_k^* = \beta_l^*$, then $r_n \{ | \beta_k^* - \beta_l^* + \frac{u_k - u_l}{r_n}| - |\beta_k^* - \beta_l^*| \} = |u_k - u_l|$, and $\lambda_n \hat w_{kl(-)}\{|\beta_k^* - \beta_l^* + \frac{u_k - u_l}{r_n}| - |\beta_k^* - \beta_l^*|\} = \lambda_n \hat c_{kl}|u_k - u_l| [ {\mathcal{V}}_{n, kl}]^{-\frac{\tau}{2}} [{\mathcal{V}}_{n, kl}^{-\frac{1}{2}}(\tilde \beta_k - \tilde \beta_l)]^{-\tau}$, where ${\mathcal{V}}_{n,kl} = [\mathcal{V}_n]_{kk} + [\mathcal{V}_n]_{ll} - 2[\mathcal{V}_n]_{kl}$ denotes the consistent variance estimator of $\tilde \beta_k - \tilde \beta_l$. By the asymptotic normality of $\tilde \bbeta_{\rm init}$, ${\mathcal{V}}_{n,kl}^{-\frac{1}{2}}(\tilde \beta_k - \tilde \beta_l) = O_p(1)$. Furthermore, because $\hat c_{kl(-)} \xrightarrow[]{P} c_{kl(-)}$, and the slowest rate at which $\tilde {\mathcal{V}}_n$ goes to 0 is $r_n^{-2}$, we have $\lambda_n \hat w_{kl(-)}\{|\beta_k^* - \beta_l^* + \frac{u_k - u_l}{r_n}| - |\beta_k^* - \beta_l^*|\} = \lambda_n \hat c_{kl}|u_k - u_l| [ {\mathcal{V}}_{n, kl}]^{-\frac{\tau}{2}} [{\mathcal{V}}_{n, kl}^{-\frac{1}{2}}(\tilde \beta_k - \tilde \beta_l)]^{-\tau} = O_p(\lambda_n r_n^{\tau - 1} c_{kl(-)} |u_k-u_l|)$. By the assumption that $\lambda_n r_n^{\tau - 1} \to \infty$ and $0 < c_{kl(-)} <\infty$, we have $\lambda_n \hat w_{kl(-)}\{|\beta_k^* - \beta_l^* + \frac{u_k - u_l}{r_n}| - |\beta_k^* - \beta_l^*|\}\to \infty$ with probability tending one unless $u_k - u_l$ goes to 0 sufficiently fast, faster than $\lambda_n r_n^{\tau - 1}$. 

By the same argument, we can show that if $\beta_k^* = -\beta_l^*$, then $\lambda_n \hat w_{kl(+)}\{|\beta_k^* + \beta_l^* + \frac{u_k + u_l}{r_n}| - |\beta_k^* + \beta_l^*|\}\to \infty$ with probability tending one unless $u_k + u_l$ goes to 0 sufficiently fast, faster than $\lambda_n r_n^{\tau - 1}$. 

By proving via contradiction (similar to the arguments in the proof of Theorem \ref{theorem: dlasso}), we can show that for those indices for which $\beta_k^* = 0$, the corresponding $\hat u_k \xrightarrow[]{P} 0$ at a rate faster than $\lambda_n r_n^{\tau - 1}$; for those pairs of indices for which $\beta_k^* = \beta_l^*$, the corresponding $\hat u_k - \hat u_l \xrightarrow[]{P} 0$ at a rate faster than $\lambda_n r_n^{\tau - 1}$; for those pairs of indices for which $\beta_k^* = - \beta_l^*$, the corresponding $\hat u_k + \hat u_l \xrightarrow[]{P} 0$ at a rate faster than $\lambda_n r_n^{\tau - 1}$. In other words, $\hat {\bm u}^{(n)}$ must have oracle structure with probability tending one as $n\to \infty$. Note that when $\bm u$ has the oracle structure, (i.e., for those indices for which $\beta_k^* = 0$, the corresponding $u_k = 0$; for those pairs of indices for which $\beta_k^* = \beta_l^*$, the corresponding $u_k = u_l$; for those pairs of indices for which $\beta_k^* = - \beta_l^*$, the corresponding $u_k = - u_l$), it is also true that $G^T C_g \bm u = \bm u$ by following almost the same argument as in the proof of Lemma 1. This implies that when $n$ is large enough, $G^TC_g \hat {\bm u}^{(n)} = \hat {\bm u}^{(n)}$ holds with high probability.

As a result, when $n$ is large enough, we can instead minimize $V^{(n)}(G^TC_g\bm u)$ for $\bm u$. Note that
\begin{align}
    V^{(n)}(G^TC_g\bm u) & = \frac{1}{r_n^2} \bm u^T C_g^TG(\hat \Pi^T W \hat \Pi - V)G^T C_g \bm u - \frac{1}{r_n} \bm u^TC_g^TG(\hat \Pi^T W \hat \bGamma - \hat \Pi^T W \hat \Pi \bbeta^* + V \bbeta^*) + \lambda_n \tilde P(G^TC_g\bm u) \nonumber \\
    & = \frac{1}{r_n^2} \bm u^T C_g^TG(\hat \Pi^T W \hat \Pi - V)G^T C_g \bm u - \frac{1}{r_n} \bm u^TC_g^TG(\hat \Pi^T W \hat \bGamma - \hat \Pi^T W \hat \Pi G^T C_g\bbeta^* + V G^T C_g\bbeta^*) \nonumber \\
    & + \lambda_n \tilde P(G^TC_g\bm u) \nonumber \\
    & = \frac{1}{r_n^2} \bm u_g^T(\hat \Pi_g^T W \hat \Pi_g - V_g) \bm u_g - \frac{1}{r_n} \bm u_g^T(\hat \Pi_g^T W \hat \bGamma - \hat \Pi_g^T W \hat \Pi_g C_g\bbeta^* + V_g C_g\bbeta^*) + \lambda_n \tilde P(G^T\bm u_g) \label{obj: grouped}
\end{align}
where we let $\bm u_g = C_g \bm u = (u_{g,1},\dots,u_{g,L})^T$ be the transformed unknown parameters. 
\eqref{obj: grouped} is essentially the objective function after collapsing and droping the elements of $\hat \bgamma_j$ according to the oracle structure. 

Additionally, note that $\tilde P(G^T\bm u_g) = O_p(\lambda_n/r_n) (\sum_{k=1}^{L}u_{g,k} {\rm sign}(\bar \beta_k^*) + \sum_{1\le k < l \le L} (u_{g,k} - u_{g,l}){\rm sign}(\bar \beta_k^* - \bar \beta_l^*) + \sum_{1\le k < l \le L}(u_{g,k}+u_{g,l}){\rm sign}(\bar \beta_k^* + \bar \beta_l^*)) = O_p(\lambda_n/r_n) f(\bm u_g)$, where $f(\bm u_g) = \bm a + \bm b^T \bm u_g$ is a linear function in $\bm u _g$, where $\bm a$ and $\bm b$ are bounded constants. Therefore, when $n$ is large enough, we only need to minimize
\begin{align*}
    \frac{1}{r_n^2} \bm u_g^T(\hat \Pi_g^T W \hat \Pi_g - V_g) \bm u_g - \frac{1}{r_n} \bm u_g^T(\hat \Pi_g^T W \hat \bGamma - \hat \Pi_g^T W \hat \Pi_g C_g\bbeta^* + V_g C_g\bbeta^*) + O_p(\frac{\lambda_n}{r_n}) f(\bm u_g)
\end{align*}
This objective is convex and the minimizer
\begin{align*}
    \hat{\bm u}_{g}^{(n)} = r_n (\hat\Pi_g^T W \hat \Pi_g-V_g)_{\mathcal{A}}^{-1}(\hat \Pi_g^T W \hat \bGamma - \hat \Pi_g^T W \hat \Pi_g C\bbeta^* + V_g C\bbeta^*) + O_p( \lambda_n/r_n)(\hat\Pi_g^T W \hat \Pi_g-V_g)^{-1} \bm b
\end{align*}

Let $\hat \bbeta_g^{(n)} = \bar \bbeta^* + \bm u_g^{(n)}/r_n$. Then we have, 
\begin{align*}
    (\hat \bbeta_g^{(n)} - \bar \bbeta^*) & = (\hat\Pi_g^T W\hat\Pi_g - V_g)^{-1} (\hat \Pi_g^T W \hat \bGamma - \hat \Pi_g^T W \hat \Pi_g C\bbeta^* + V_g C\bbeta^*)  + O_p( \lambda_n) (\hat\Pi_g^T W \hat \Pi_g-V_g)^{-1}\bm b,
\end{align*}
and,
\begin{align*}
    \mathbb{V}_g^{-\frac{1}{2}} (\Pi_g^T W \Pi_g)(\hat \bbeta_{g}^{(n)} - \bar \bbeta^*) & = \mathbb{V}_{g}^{-\frac{1}{2}} (\Pi_g^T W \Pi_g) (\hat\Pi_g^T W\hat\Pi_g - V)^{-1} (\hat \Pi_g^T W \hat \bGamma - \hat \Pi_g^T W \hat \Pi_gC\bbeta^* + V_g C\bbeta^*)  \\
    & + O_p( \lambda_n) \mathbb{V}_g^{-\frac{1}{2}} (\Pi_g^T W \Pi_g)(\hat\Pi_g^T W \hat \Pi_g-V_g)^{-1}\bm b \\
    & = \underbrace{\mathbb{V}_{g}^{-\frac{1}{2}} (\Pi_g^T W \Pi_g)_{\mathcal{A}} (\hat\Pi_g^T W\hat\Pi_g - V_g)^{-1} \mathbb{V}_{g}^{\frac{1}{2}}}_{A_1} \underbrace{ \mathbb{V}_{g}^{-\frac{1}{2}} (\hat \Pi_g^T W \hat \bGamma - \hat \Pi_g^T W \hat \Pi_g C\bbeta^* + V_g C \bbeta^*)}_{A_2}\\
    & + O_p( \lambda_n) \underbrace{\mathbb{V}_{g}^{-\frac{1}{2}} (\Pi_g^T W \Pi_g) (\hat\Pi_g^T W \hat \Pi_g-V)^{-1}\mathbb{V}_{g}^{\frac{1}{2}}}_{A_1} \mathbb{V}_{g}^{-\frac{1}{2}}\bm b
\end{align*}
By the result S10 and Lemma 6 in \cite{wu2024more}, we know $A_1 \xrightarrow[]{P} I_{L}$ and $A_2 \xrightarrow[]{D} N(\bm 0, I_{L})$. Additionally, by Lemma 2 in the SRIVW paper, $\mathbb{V}_{\mathcal{A}}^{-\frac{1}{2}}$ goes to 0 at a rate at least $\sqrt{\mu_{n,\min} + p}$. Thus, the second term is of order $O_p(\lambda_n/\sqrt{\mu_{n,\min}+p}) = o_p(1)$ as $\lambda_n/r_n \rightarrow 0$. Then, by the Slutsky's theorem, we show the asymptotic normality of $\hat \bbeta_g^{(n)}$. Finally, note that $\hat \bbeta_g^{(n)} - \bar \bbeta^* = \bm u_g^{(n)}/r_n = C_g \bm u^{(n)}/r_n = C_g(\hat \bbeta - \bbeta^*)$ by the definition of $\bm u_g$ and $\bm u$. We complete the proof of asymptotic normality.

We now prove the consistency in risk factor selection and group identification. Note that the asymptotic normality result implies that for $\forall k \in \mathcal{B}$, $\probP(k \in \mathcal{B}_n) \to 1$. It remains to show that $\forall k^{'}  \in \mathcal{B}^c$, $\probP(k^{'} \in B_n)\to 0$. For $k^{'} \in \mathcal{B}^c$, there are three possibility. First, there exists an index $m \in \{1,\dots,K\}$ such that $\beta_m^* = 0$ but $\hat \beta_m \neq 0$. Second, there exists a pair of indices $(m, l)$ for $m,l\in\{1,\dots,K\}$ such that $\beta_m^* = \beta_l^*$ but $\hat \beta_m \neq \hat \beta_l$. Third, there exists a pair of indices $(m, l)$ such that $\beta_m^* = - \beta_l^*$ but $\hat \beta_m \neq -\hat \beta_l$. In the following, we show by contradiction that these three scenarios happen with probability tending 0.

For the first scenario, Since $\hat \beta_m \neq 0$, by the KKT conditions, we know $\hat \beta_m \neq 0$ satisfies
\begin{align}
    \bigg [(\hat \Pi^T W \hat \Pi - V) \hat \bbeta^{(n)} - \hat \Pi^T W \hat \bGamma\bigg]_{m} & = \lambda_n \{\hat w_{m}{\rm sign}(\hat \beta_m) + \sum_{1\le j < m} \hat w_{jm(-)}{\rm sign}(\hat \beta_m - \hat \beta_j) + \sum_{m < k \le K} \hat w_{mk(-)}{\rm sign}(\hat \beta_k - \hat \beta_m) \nonumber \\
    & + \sum_{1\le j < m} \hat w_{jm(+)}{\rm sign}(\hat \beta_m + \hat \beta_j) + \sum_{m < k \le K} \hat w_{mk(+)} {\rm sign} (\hat \beta_m + \hat \beta_k  \}\label{KKT eq pacs}
\end{align}
Following the same argument in the proof of Theorem \ref{theorem: dlasso}, we can show that the left hand side is\\ $O_p(\mu_{n,\max}/r_n)$ using the asymptotic normality result. Now we focus on the right hand side. Since $\beta_m = 0$, $[\mathcal{V}_n]_{mm}^{-\frac{1}{2}}\tilde \beta_m = O_p(1)$. Hence, $\hat w_m = [\mathcal{V}_n]_{mm}^{-\frac{\tau}{2}} ([\mathcal{V}_n]_{mm}^{-\frac{1}{2}}\tilde \beta_m)^{-\tau} = O_p(r_n^{\tau})$. Consider $\hat w_{jm(-)} = \hat c_{jm(-)} (|\tilde \beta_m - \tilde \beta_j|)^{-\tau}$ for $j < m$. If $\beta_j = 0$, then $ {\mathcal{V}}_{n, jm}^{-\frac{1}{2}}(\tilde \beta_m - \tilde \beta_j) = O_p(1)$, where $ {\mathcal{V}}_{n, jm} = [\mathcal{V}_n]_{kk} + [\mathcal{V}_n]_{ll} - 2 [\mathcal{V}_n]_{kl}$ denotes the consistent variance estimator of $\tilde \beta_m - \tilde \beta_j$, and thus $\hat w_{jm(-)} = O_p(r_n^{\tau})$. If $\beta_j \neq 0$, then $|\tilde \beta_m - \tilde \beta_j| \xrightarrow[]{P} |\beta_j|$ and  $\hat w_{jm(-)} = O_p(1)$. Hence, $\hat w_{jm(-)} = O_p(r_n^{\tau})$ for $j < m$. Similarly, we can argue that $\hat w_{mk(-)}$ for $k > m$, $\hat w_{jm(+)}$ for $j < m$, and $\hat w_{mk(+)}$ for $k > m$ are all at most $O_p(r_n^{\tau})$. Therefore, the right hand side is at most $O_p(\lambda_n r_n^{\tau})$. By the assumption that $\lambda_n r_n^{\tau + 1}/\mu_{n,\max} \to \infty$, we have a contradiction that \eqref{KKT eq pacs} cannot hold for large enough $n$. 

For the second scenario, we only need to consider the case when $\beta_m = \beta_l \neq 0$, because the proof for the case when $\beta_m = \beta_l = 0$ follows directly from the proof in the first scenario. We again analyze the right hand side of \eqref{KKT eq pacs}. For the first term, because $\beta_m^* \neq 0$, $|\tilde \beta_m| \xrightarrow[]{P} \beta_m^*$ and hence $\hat w_m = O_p(1)$. Then, consider $\hat w_{jm(-)} = \hat c_{jm(-)} (|\tilde \beta_m - \tilde \beta_j|)^{-\tau}$ for $j < m$ and $\hat w_{mk(-)} = \hat c_{mk(-)} (|\tilde \beta_k - \tilde \beta_m|)^{-\tau}$ for $m < k$. For the indices $l$ such that $\beta_l \neq \beta_m$, we know $|\tilde \beta_m - \tilde \beta_l| \xrightarrow[]{P} |\beta_m - \beta_l|$ and thus $\hat w_{lm} = O_p(1)$. For the indices $l$ such that $\beta_l = \beta_m \neq 0$, we know  $ {\mathcal{V}}_{n, lm}^{-\frac{1}{2}}(\hat \beta_m - \hat \beta_l) = O_p(1)$ where $ {\mathcal{V}}_{n, lm}$ denotes the consistent variance estimator of $\tilde \beta_m - \tilde \beta_l$. By the definition of this scenario, there exists at least one such an index $l$, and additionally, $\hat \beta_m \neq \hat \beta_l$. Therefore, we must have $\hat w_{lm(-)}{\rm sign}(\hat \beta_m - \hat \beta_l) = O_p(r_n^{\tau})$. For the terms $\hat w_{jm(+)}$ and $\hat w_{mk(+)}$, following similar arguments, we know they are at most $O_p(r_n^{\tau})$. Therefore, the right hand side is $O_p(\lambda_n r_n^{\tau})$. Again, by the assumption that $\lambda_n r_n^{\tau + 1}/\mu_{n,\max} \to \infty$, we have a contradiction that \eqref{KKT eq pacs} cannot hold for large enough $n$. 

The proof for the third scenario mirrors the proof of the second scenario. We can follow similary arguments to show for the pair of indices $(m, l)$ such that $\beta_m^* = - \beta_l^*  \neq 0$ but $\hat \beta_m \neq \hat \beta_l$, $\hat w_{ml(+)}{\rm sign}(\hat \beta_m + \hat \beta_l) = O_p(r_n^{\tau})$. Thus, the right hand side of \eqref{KKT eq pacs} is still $O_p(\lambda_n r_n ^{\tau})$ and hence we still have a contradiction. 

\end{proof}

\end{document}